%% file: QPL-2025-47.tex
\documentclass[submission,copyright, creativecommons]{eptcs}
\providecommand{\event}{QPL 2025}
\usepackage[utf8]{inputenc}
\usepackage[english]{babel}
\usepackage[T1]{fontenc}
\usepackage{tikz-cd}
\usepackage{tikzit}
\usepackage{amsmath}
\usepackage{amsthm}
\usepackage{amssymb}
\usepackage{braket}
\usepackage{appendix}
\usepackage{hyperref}
\usepackage{float}
\usepackage{stmaryrd} 
\usepackage{mathtools}
\usepackage{crossreftools}

\usepackage{multicol} 
\usepackage{underscore} 

\usepackage[shortlabels, inline]{enumitem}
\setlist{nosep} 

\usepackage{xspace}

\newcommand{\symd}{\mathbin{\Delta}\xspace}

\newcommand{\symdi}[1]{{\scalebox{1.5}{$\symd$}}_{#1}\,}

\makeatletter
\newcommand\etc{etc\@ifnextchar.{}{.\@}\xspace}
\newcommand\ie{i.e.\@\xspace}  
\newcommand\eg{e.g.\@\xspace}
\makeatother

\newcommand{\odd}[1]{\mathsf{Odd}\left(#1\right)}

\newcommand{\codd}[1]{\mathsf{Odd}\left[#1\right]}
\newcommand{\odds}[2]{\mathsf{Odd}_{#1}\left(#2\right)}

\newcommand{\codds}[2]{\mathsf{Odd}_{#1}\left[#2\right]}

\newcommand{\pow}[1]{\ensuremath{2^{ #1 }}}
\newcommand{\trl}{\triangleleft} 
\newcommand{\Wpred}{\mathcal{W}_{pred}}
\newcommand{\Vsucc}{\mathcal{V}_{succ}}

\newcommand{\YZ}{\normalfont YZ\xspace}

\newcommand{\LOG}{labelled open graph}

\newcommand{\ketbra}[2]{\ensuremath{\ket{#1}\!\bra{#2}}}
\newcommand{\abs}[1]{\ensuremath{\left| #1 \right|}}

\newcommand{\someset}{\mathcal{A}}
\newcommand{\otherset}{\mathcal{D}}
\newcommand{\bigO}{\mathcal{O}}

\newcommand{\triplecase}[6]{
    #1 \quad & \text{and} \quad #2, \text{or}\\
    #3 \quad & \text{and} \quad #4, \text{or}\\
    #5 \quad & \text{and} \quad #6.
}

\newcommand{\A}{{A}}

\newcommand{\I}{Id}

\newcommand{\Xlike}{\mathcal{X}}
\newcommand{\Zlike}{\mathcal{Z}}

\newcommand{\planar}{\mathcal{L}}

\input{zx.tikzstyles}
\theoremstyle{definition}
\newtheorem{theorem}{Theorem}[section]
\newtheorem{lemma}[theorem]{Lemma}
\newtheorem{proposition}[theorem]{Proposition}
\newtheorem{corollary}[theorem]{Corollary}
\newtheorem{observation}[theorem]{Observation}
\newtheorem{definition}[theorem]{Definition}
\newtheorem{example}[theorem]{Example}
\newtheorem{remark}[theorem]{Remark}

\def\abs#1{\left| #1 \right|}

\newcommand{\comp}[1]{\bar{#1}} 
\newcommand{\ld}{\lambda}
\newcommand{\sse}{\subseteq}
\newcommand{\pmm}[1]{\begin{pmatrix} #1 \end{pmatrix}}
\newcommand{\smm}[1]{\left(\begin{smallmatrix} #1 \end{smallmatrix}\right)}

\renewcommand{\v}[1]{\mathbf{#1}}

\usepackage{pgfplots}
\usetikzlibrary{decorations.markings,decorations.pathmorphing, decorations.text,positioning}
\usetikzlibrary{backgrounds,folding}
\usetikzlibrary{shapes.geometric, shapes.misc}
\usetikzlibrary{arrows.meta}
\pgfdeclarelayer{edgelayer}
\pgfdeclarelayer{nodelayer}
\pgfsetlayers{background,edgelayer,nodelayer,main}
\tikzstyle{every picture}=[baseline=-0.25em]
\tikzset{every path/.style={draw=black!80, line width=0.7pt}}
\tikzstyle{none}=[inner sep=0mm]
\tikzset{
	tickedge/.style={
		decoration={ markings,
			mark=at position .5 with {\draw (0,2pt) -- (0,-2pt);}
		},
		postaction={decorate}
	},
}
\usepackage{etoolbox}

\newcommand{\interp}[1]{\left\llbracket #1 \right\rrbracket}

\newcommand{\eq}[2][]{
	#1
	\underset{\substack{#2}}{=}
	#1
}

\newcommand{\DtoD}[2][]{
	#1
	\underset{\substack{#2}}{\to}
	#1
}

\title{Inserting Planar-Measured Qubits into MBQC Patterns while Preserving Flow}
\author{Miriam Backens
\institute{Université de Lorraine, CNRS, Inria, LORIA, F-54000 Nancy, France}
\and
Thomas Perez
\institute{Inria, Palaiseau, France\\
LIX, CNRS, Ecole Polytechnique, Institut Polytechnique de Paris, Palaiseau, France\\
CPHT, CNRS, Ecole Polytechnique, Institut Polytechnique de Paris, Palaiseau, France
}
}
\def\titlerunning{Inserting Planar-Measured Qubits into MBQC Patterns while Preserving Flow}
\def\authorrunning{M.~Backens \& T.~Perez}
\begin{document}
\maketitle

\begin{abstract}
 In the one-way model of measurement-based quantum computation (MBQC), computation proceeds via single-qubit measurements on a resource state.
 Flow conditions ensure that the overall computation is deterministic in a suitable sense, and are required for efficient translation into quantum circuits.
 Procedures that rewrite MBQC patterns -- e.g. for optimisation, or adapting to hardware constraints -- thus need to preserve the existence of flow.
 Most previous work has focused on rewrites that reduce the number of qubits in the computation, or that introduce new Pauli-measured qubits.

 Here, we consider the insertion of planar-measured qubits into MBQC patterns, \ie arbitrary measurements in a plane of the Bloch sphere spanned by a pair of Pauli operators; such measurements are necessary for universal MBQC.
 We extend the definition of causal flow, previously restricted to $XY$-measurements only, to also permit $YZ$-measurements and derive the conditions under which a $YZ$-insertion preserves causal flow.
 Then we derive conditions for $YZ$-insertion into patterns with gflow or Pauli flow, in which case the argument straightforwardly extends to $XZ$-insertions as well.
 We also show that the `vertex splitting' or `neighbour unfusion' rule previously used in the literature can be derived from $YZ$-insertion and pivoting.
 This work contributes to understanding the broad properties of flow-preserving rewriting in MBQC and in the ZX-calculus more broadly, and it will enable more efficient optimisation, obfuscation, or routing.
\end{abstract}

\section{Introduction}

The one-way model of measurement-based quantum computing (MBQC) is a universal model of quantum computation driven by single-qubit measurements on a graph state \cite{raussendorfOneWayQuantumComputer2001}.
Individual measurements are probabilistic, yet computations can be deterministic overall by modifying later measurements depending on the outcomes of earlier ones.
MBQC is a promising model for physically building quantum computers (\eg based on photonics \cite{takedaToward2019,zilkCompiler2022,deFeliceFusion2024}) as well as for secure delegated quantum computation \cite{broadbentUniversalBlindQuantum2009,kapourniotisUnifying2024}.
It also has theoretical applications in quantum secret-sharing protocols \cite{kashefiInformation2009} and verifying quantum computations \cite{gheorghiuVerification2019,kapourniotisUnifying2024}.
The higher flexibility of MBQC compared to quantum circuits means that it is easier to optimise computations in the one-way model, \eg to trade depth against number of ancillas \cite{broadbentParallelizingQuantumCircuits2009}, reduce the number of non-stabiliser operations \cite{duncanGraphtheoreticSimplificationQuantum2020}, or the number of entangling operations \cite{staudacherReducing2QuBitGate2023}.

Many of these applications involve translating a computation given as a quantum circuit into MBQC and then back to a circuit \cite{duncanGraphtheoreticSimplificationQuantum2020,backensThereBackAgain2021,mcelvanneyCompleteFlowPreservingRewrite2023,mcelvanneyFlowpreservingZXcalculusRewrite2023} using the ZX-calculus, a diagrammatic language that can represent both quantum circuits over the CNOT and single-qubit gate set, and MBQC \cite{coeckeInteractingQuantumObservables2011}.
The ZX-calculus also comes with multiple complete equational theories, that allow purely graphical reasoning for translation between the two models as well as for optimisation or other applications \cite{vandeweteringZXcalculusWorkingQuantum2020}.
Yet the existing axioms do not always preserve the `MBQC-form', nor the property of being robustly deterministic.

The widely-used notion of \emph{robust determinism} is captured by a family of flow properties that depend on the graph state underlying the MBQC, as well as the types of measurements, which can be `Pauli' or `planar' (see Section~\ref{s:preli}) \cite{browneGeneralizedFlowDeterminism2007,mhallaCharacterisingDeterminismMBQCs2022}.
This information is also captured in the ZX-diagram describing a one-way computation, so it is meaningful to consider `flow-preserving rewrite rules' for MBQC-form ZX-diagrams.
This research is of broader importance, as there exist polynomial-time algorithms for translating ZX-diagrams with flow into quantum circuits \cite{duncanGraphtheoreticSimplificationQuantum2020,backensThereBackAgain2021,simmonsRelatingMeasurementPatterns2021} for physical implementation, whereas translation to circuits is \#P-hard in general \cite{debeaudrapCircuitExtractionZXdiagrams2022}.
A complete set of flow-preserving rewrite rules is known for the Clifford fragment: one-way computations containing only Pauli measurements; this includes the insertion of Pauli-$Z$ measured qubits \cite{mcelvanneyCompleteFlowPreservingRewrite2023}.
Other flow-preserving rewrite rules either reduce or leave invariant the number of qubits \cite{duncanGraphtheoreticSimplificationQuantum2020,backensThereBackAgain2021,holkerCausal2023}.
Additionally, a version of `vertex splitting' (where one qubit is replaced by a chain of three qubits instead) is known to preserve Pauli flow \cite{mcelvanneyFlowpreservingZXcalculusRewrite2023}.

Here, we look more broadly at the idea of inserting new planar-measured qubits into one-way computations.
We first straightforwardly extend the definition of causal flow (the most restricted type of flow, which forces MBQC to be most `circuit-like') to allow measurements in the $YZ$-plane as well as the standard $XY$-plane\footnote{This extension was also suggested independently by Calum Holker \cite{holkerCausal2023}.}, and then derive the conditions under which $YZ$-measurements can be added to a computation with causal flow.
Next, we derive similar conditions for the insertion of $YZ$-measured qubits into MBQC with the more general gflow and Pauli flow: the algebraic formulation of focused Pauli flow \cite{mitosekAlgebraic2024} allows both cases to be treated simultaneously.
A simple parity change allows those conditions to be adapted to $XZ$-measurements instead.
Finally, we combine $YZ$-insertions with the flow-preserving pivot operation \cite{duncanGraphtheoreticSimplificationQuantum2020,backensThereBackAgain2021,simmonsRelatingMeasurementPatterns2021} (which can change measurement labels) to generalise the vertex splitting rule: this now effectively allows the insertion of pairs of $XY$-measured qubits under certain conditions.
Such insertions of planar-measured vertices instead of Pauli-measured ones are useful if one wishes to stay within the setting of gflow (instead of Pauli flow), which is sometimes desirable.
They also have applications in generating variational ans\"atze in a measurement- or ZX-based setting such as \cite{fergusonMeasurement-based2021,ewenApplication2024}.

In the following, we give the preliminary definitions and existing results in Section~\ref{s:preli}.
$YZ$-insertion for MBQC with causal flow is analysed in Section~\ref{s:causal}, and for Pauli flow in Section~\ref{s:gflow}. Finally, insertion of measurements in other planes is considered in Section~\ref{s:other-labels}, followed by the conclusions in Section~\ref{s:conclusions}.

\section{Preliminaries}\label{s:preli}

We begin by formalising the graph-theoretic notions and related results that will be used later.

\begin{definition}
 A \emph{labelled open graph} is a tuple $(G,I,O,\ld)$ consisting of a simple graph $G=(V,E)$, subsets $I,O\sse V$ called the \emph{input} and \emph{output vertices}, and a function $\ld:\comp{O}\to\{X,Y,Z,XY,XZ,YZ\}$ called the \emph{measurement labelling}.
\end{definition}

Given a labelled open graph, let $\Xlike := \{v\in V\setminus(I\cup O)\mid \ld(v)\in\{XY,X,Y\}\}$ be the set of `$X$-like' internal vertices and $\Zlike := \{v\in V\setminus(I\cup O)\mid \ld(v)\in\{XZ,YZ,Z\}\}$ the set of `$Z$-like' internal vertices. Define $\planar := \{v\in V\setminus(I\cup O)\mid \ld(v)\in\{XY,XZ,YZ\}\}$ to be the set of planar measured internal vertices.

Throughout this paper, we denote by $\symd$ the symmetric difference of two sets, which consists of those elements contained in exactly one of the original sets: \ie $\someset\symd\otherset = (\someset\cup\otherset)\setminus(\someset\cap\otherset)$.
The operation is associative so it straightforwardly extends to a family of sets: the symmetric difference $\symdi{1\leq k\leq n} \someset_k$ contains exactly those elements that appear in an odd number of the constituent sets $\someset_k$.

\begin{definition}
 Let $G = (V,E)$ have vertices $V$ and edges $E$.
 For any vertex $u \in V$, we denote by $N_G(u) := \{v \in V \mid \{u,v\}\in E\}$ the \emph{neighbourhood of $u$} and by $N_G[u] := N_G(u)\symd\{u\}$ the \emph{closed neighbourhood}.
 Similarly, for any subset $\someset\sse V$, we denote by $\odds{G}{\someset} := \symdi{v\in\someset} N_G(v)$ the \emph{odd neighbourhood of $\someset$} and by $\codds{G}{\someset} := \symdi{v\in\someset} N_G[v] = \odds{G}{\someset} \symd \someset$ the \emph{closed odd neighbourhood of $\someset$}.
 The subscript designating the graph may be left out if no confusion is likely to result.
\end{definition}

\begin{definition}
 Let $u\in V$, then the graph $G\star u$ resulting from a \emph{local complementation about $u$} is $G\star u := (V, E \symd \{\{v,w\}\mid v,w\in N_G(u), v\neq w\})$.
 Let $u\in V$ and $v\in N_G(u)$, then the graph $G\wedge uv$ resulting from a \emph{pivot about the edge $\{u,v\}$} is $G\wedge uv := G\star u \star v \star u$, \ie it arises by a sequence of local complementations.
\end{definition}

Labelling the pivot by the edge $\{u,v\}$ is well-defined as $G\star u \star v \star u = G \star v \star u \star v$~\cite{bouchetGraphic1988}.

\subsection{Flow properties}

Having established the relevant graph-theoretic concepts, we are now ready to define flow properties.

\begin{definition}[Causal flow {\cite[Definition~2]{danosDeterminismOnewayModel2006}}]\label{def:causal-flow}
 A \emph{causal flow} for a \LOG\ $\Gamma = (G,I,O,\lambda)$ is a pair $(c,\prec)$ where $c : \overline{O} \rightarrow \overline{I}$ and $\prec$ is a strict partial order, such that for all $u\in V$:

 ${\scriptstyle\bullet}$ $c(u)\in N(u)$ \hfill ${\scriptstyle\bullet}$ $u \prec c(u)$ \hfill ${\scriptstyle\bullet}$ $\forall v \in N(c(u)), u \neq v \Rightarrow u \prec v$. \hfill\;
\end{definition}

The presence of causal flow implies that a one-way computation can be implemented robustly deterministically, but it is not necessary \cite{browneGeneralizedFlowDeterminism2007}.
For labelled open graphs (\ie computations where no partial order is specified \textit{a priori}), the following property is both necessary and sufficient \cite{browneGeneralizedFlowDeterminism2007,mhallaCharacterisingDeterminismMBQCs2022}.

\begin{definition}[Pauli flow {\cite[Definition~5]{browneGeneralizedFlowDeterminism2007}}]\label{def:Pauli-flow}
    A \emph{Pauli flow} for a labelled open graph $(G,I,O,\lambda)$ is a pair $(c,\prec)$, where $c:\comp{O}\to\pow{\comp{I}}$ and $\prec$ is a strict partial order on $\comp{O}$, such that for all $u \in \comp{O}$:
    \begin{multicols}{2}
    \begin{enumerate}
        \item[{\crtcrossreflabel{(P1)}[P1]}] $\forall v \in c(u) . u \ne v \wedge \lambda(v) \notin \{ X, Y \} \Rightarrow u \prec v$
        \item[{\crtcrossreflabel{(P2)}[P2]}] $\forall v \in \odd{c(u)} . u \ne v \wedge \lambda(v) \notin \{ Y, Z \} \Rightarrow u \prec v$
        \item[{\crtcrossreflabel{(P3)}[P3]}] $\forall v \in \comp{O} . \neg (u \prec v) \wedge u \ne v \wedge \lambda(v) = Y \Rightarrow v\notin\codd{c(u)}$
        \item[{\crtcrossreflabel{(P4)}[P4]}] $\lambda(u) = XY \Rightarrow u \notin c(u) \wedge u \in \odd{c(u)}$
        \item[{\crtcrossreflabel{(P5)}[P5]}] $\lambda(u) = XZ \Rightarrow u \in c(u) \wedge u \in \odd{c(u)}$
        \item[{\crtcrossreflabel{(P6)}[P6]}] $\lambda(u) = YZ \Rightarrow u \in c(u) \wedge u \notin \odd{c(u)}$
        \item[{\crtcrossreflabel{(P7)}[P7]}] $\lambda(u) = X \Rightarrow u \in \odd{c(u)}$
        \item[{\crtcrossreflabel{(P8)}[P8]}] $\lambda(u) = Z \Rightarrow u \in c(u)$
        \item[{\crtcrossreflabel{(P9)}[P9]}] $\lambda(u) = Y \Rightarrow u \in \codd{c(u)}$.
    \end{enumerate}
    \end{multicols}
    The sets $c(v)$ for $v \in \comp{O}$ are called the \textit{correction sets} and $c$ is called the \emph{correction function}.
\end{definition}

We will sometimes use the term \emph{correction function} to refer to any function $c:\comp{O}\to\pow{\comp{I}}$ on some given \LOG\ $(G,I,O,\ld)$ which satisfies properties \ref{P4}--\ref{P9}.

A Pauli flow on a \LOG\ in which all measurements are planar is called a gflow\footnote{In the literature, the term `gflow' is sometimes used to refer specifically to a Pauli flow in which all measurement labels are $XY$, in which case the property allowing $XZ$ and $YZ$ measurements may be called `extended gflow'.} and has been widely studied in its own right.
Both gflow and Pauli flow are not generally unique; thus a more restricted version satisfying the following `focusing conditions' is often useful to work with.
Focusing was defined first for $XY$-only gflow \cite[Definition~3.1]{mhallaWhichGraphStates2014} before being extended \cite[Proposition~3.14]{backensThereBackAgain2021}.

\begin{definition}[{\cite[Definition~4.3]{simmonsRelatingMeasurementPatterns2021}}]\label{def:focused}
 Given $(G,I,O,\ld)$, a set $\someset\sse\comp{I}$ is \emph{focused over $S\sse\comp{O}$} if:
 \begin{multicols}{2}
 \begin{enumerate}
  \item[{\crtcrossreflabel{(F1)}[F1]}] $\forall w \in S \cap \someset. \lambda(w) \in \{ XY, X, Y \}$,
  \item[{\crtcrossreflabel{(F2)}[F2]}] $\forall w \in S \cap \odd{\someset} . \lambda(w) \in \{ XZ, YZ, Y, Z \}$,
  \item[{\crtcrossreflabel{(F3)}[F3]}] $\forall w \in S . \lambda(w) = Y \Rightarrow w \notin \codd{\someset}$.
 \end{enumerate}
 \end{multicols}
 A \emph{focused set} $\someset$ for $\Gamma$ is focused over $\comp{O}$.
 A correction function $c$ is \emph{focused} if $c(v)$ is focused over $\comp{O}\setminus\{v\}$ for all $v\in\comp{O}$.
 A Pauli flow is called \emph{focused} if its correction function is focused.
\end{definition}

\begin{lemma}[{\cite[Lemma~4.6]{simmonsRelatingMeasurementPatterns2021}}]\label{lem:focused-Pauli-flow}
    For any labelled open graph, if a Pauli flow exists, then there also exists a focused Pauli flow.
\end{lemma}

\subsection{Algebraic Pauli flow}

A major advantage of focused Pauli flow is that it is amenable to an algebraic characterisation, replacing conditions \ref{P1}--\ref{P9} and \ref{F1}--\ref{F3} by two conditions on certain matrices derived from the \LOG.
First, we define an `extended adjacency matrix' to be the adjacency matrix of the graph that results from some \LOG\ $(G,I,O,\ld)$ by adding self-loops to all vertices measured $Y$ or $XZ$.
This convention will make the following definitions simpler.

\begin{definition}[{\cite[Definition~3.16]{mitosekAlgebraic2024}}]
  Let $\Gamma=(G,I,O,\ld)$ be a labelled open graph.
    Then the \emph{extended adjacency matrix} $\A$ of $\Gamma$ satisfies $\A_{v,w} = 1$ if and only if either $\{v,w\}\in E$ or $v=w \wedge \ld(v)\in\{Y,XZ\}$.
\end{definition}

\begin{definition}[{\cite[Definition~3.4]{mitosekAlgebraic2024}}]\label{def:flow-demand}
    Let $\Gamma = (G,I,O,\lambda)$ be a labelled open graph with extended adjacency matrix $\A$.
    The \emph{flow-demand matrix} $M_{\Gamma}$ is the $(n-n_O) \times (n-n_I)$ matrix with rows corresponding to non-outputs and columns corresponding to non-inputs, whose $v$-labelled row satisfies the following properties:
    if $\lambda(v) \in \{ X, Y, XY \}$, then $M_{v,w} = \A_{v,w}$ for all $w\in\comp{I}$, and
    if $\lambda(v) \in \{Z, YZ, XZ\}$, then $M_{v,v} = 1$ and $M_{v,w} = 0$ for all $w\in\comp{I}\setminus\{v\}$.
\end{definition}

\begin{definition}[{\cite[Definition~3.5]{mitosekAlgebraic2024}}]\label{def:order-demand}
     Let $\Gamma = (G,I,O,\lambda)$ be a labelled open graph with extended adjacency matrix $\A$.
     The \emph{order-demand matrix} $N_{\Gamma}$ is the $(n-n_O) \times (n-n_I)$ matrix with rows corresponding to non-outputs and columns corresponding to non-inputs, whose $v$-labelled row satisfies the following properties:
     if $\lambda(v) \in \{ X, Y, Z \}$, then $N_{v,w} = 0$ for all $w\in\comp{I}$;
     if $\lambda(v) \in \{ XZ, YZ \}$, then $N_{v,w} = \A_{v,w}$ for all $w\in\comp{I}$; and
     if $\lambda(v) = XY$, then $N_{v,v} = 1$ (provided that the $v$ column exists) and $N_{v,w} = 0$ for all $w\in\comp{I}\setminus\{v\}$.
\end{definition}

If we represent a set of vertices $\someset\sse\comp{I}$ as a column vector, which contains a 1 in the row labelled $v$ if and only if $v\in\someset$, and multiply this vector by the flow-demand matrix, then the result encodes specific relationships between each non-output vertex $u$ and the set $\someset$ that depend on the measurement label $\ld(u)$ of the vertex under consideration.

\begin{lemma}[{\cite[Lemma~3.9]{mitosekAlgebraic2024}}]\label{lem:row-by-col-meaning}
    Let $(G,I,O,\lambda)$ be a labelled open graph.
    Let $\someset \subseteq \comp{I}$ and let $\v{a}$ be the indicator vector for this set.
    Let $u \in \comp{O}$.
    Then the product of the $u$-th row of $M$ with the indicator vector of $\someset$ satisfies $M_{u,*}\v{a} = \left(M\v{a}\right)_u = 1$ if and only if:
    $\lambda(u) \in \{ X, XY \}$ and $u \in \odd{\someset}$, or
    $\lambda(u) \in \{ XZ, YZ, Z \}$ and $u \in \someset$, or
    $\lambda(u) = Y$ and $u \in \codd{\someset}$.
\end{lemma}

In addition to the flow-demand and order-demand matrix, we now also encode the correction function as a matrix by assembling the column indicator vectors for each correction set.
That definition then enables the statement of the algebraic formulation of Pauli flow theorem.

\begin{definition}[{\cite[Definition~3.6]{mitosekAlgebraic2024}}]\label{def:correction-matrix}
    Let $\Gamma = (G,I,O,\lambda)$ be a labelled open graph and let $c$ be a correction function on $\Gamma$.
    The \textit{correction matrix $C$} encoding the function $c$ is the $(n-n_I) \times (n-n_O)$ matrix
    where $C_{u,v} = 1$ if and only if $u \in c(v)$.
    In other words, the $v$-labelled column of $C$ encodes the set $c(v)$.
\end{definition}

\begin{theorem}[Algebraic formulation of Pauli flow {\cite[Theorem~3.1]{mitosekAlgebraic2024}}]\label{thm:algebraic-Pauli}
 Let $\Gamma = (G,I,O,\lambda)$ be a labelled open graph, $c$ a correction function on $\Gamma$, $M$ the flow-demand matrix of $\Gamma$, and $N$ the order-demand matrix of $\Gamma$.
 Then there exists a strict partial order $\prec$ such that $(c,\prec)$ is a focused Pauli flow on $\Gamma$ if and only if
 $M_{\Gamma}C = \I_{\comp{O}}$ and
 $N_{\Gamma}C$ is the adjacency matrix of a directed acyclic graph,
 where $C$ is the correction matrix for $c$.
\end{theorem}

\subsection{Further properties of Pauli flow}
\label{s:further-properties}

We will not need the following definition and lemma individually, but together they yield a useful observation that will simplify handling changes to the partial order of a Pauli flow that result from vertex insertions in later sections.

\begin{definition}[{\cite[Definition~2.10]{mitosekAlgebraic2024}}]\label{def:trl_c}
    Let $\Gamma = (G,I,O,\lambda)$ be a labelled open graph and let $c$ be a correction function on $\Gamma$.
    The \emph{induced relation} $\trl_c$ is the minimal relation on $\comp{O}$ implied by \ref{P1}, \ref{P2}, \ref{P3}.
    That is, $u \trl_c v$ if and only if at least one of the following holds:
    $v \in c(u) \wedge u \ne v \wedge \lambda(v) \notin \{ X, Y \}$ (corresponding to \ref{P1}), $v \in \odd{c(u)} \wedge u \ne v \wedge \lambda(v) \notin \{ Y, Z \}$ (corresponding to \ref{P2}), and $v \in \codd{c(u)} \wedge u \ne v \wedge \lambda(v) = Y$ (corresponding to \ref{P3}).
    If the transitive closure of $\trl_c$ is a partial order, denote it by $\prec_c$ and call it the \emph{induced partial order of $c$}.
\end{definition}

\begin{lemma}[{\cite[Lemma~3.12]{mitosekAlgebraic2024}}]\label{lem:NC-induced-order}
    Let $\Gamma = (G,I,O,\lambda)$ be a labelled open graph and let $c$ be a focused correction function on $\Gamma$.
    Then, for any $u, v \in \comp{O}$ such that $u \ne v$, we have $(NC)_{u,v} = 1$ if and only if $v \trl_c u$.
    In other words, the off-diagonal part of the matrix $NC$ encodes the relation $\trl_c$.
\end{lemma}

Combining Definition~\ref{def:trl_c} and Lemma~\ref{lem:NC-induced-order}, we find the following.

\begin{observation}\label{obs:NC-interpretation}
 Let $\Gamma = (G,I,O,\lambda)$ be a labelled open graph and let $C$ be the matrix of a focused correction function $c$ on $\Gamma$.
 Then, for any $u, v \in \comp{O}$ such that $u \ne v$, we have $(NC)_{u,v} = 1$ if and only if at least one of the following holds:

 ${\scriptstyle\bullet}$ $v \in c(u) \wedge \lambda(v) \notin \{ X, Y \}$, \hfill
 ${\scriptstyle\bullet}$ $v \in \odd{c(u)} \wedge \lambda(v) \notin \{ Y, Z \}$, \hfill
 ${\scriptstyle\bullet}$ $v \in \codd{c(u)} \wedge \lambda(v) = Y$. \hfill\;

 \noindent Using the property that $c$ is focused (\ie Definition~\ref{def:focused} holds), the above conditions reduce to $v\in c(u)$ and $\ld(v)=XY$, or $v\in\odd{c(u)}$ and $\ld(v)\in\{XZ,YZ\}$.
\end{observation}

The induced partial order of Definition~\ref{def:trl_c} plays an important role: if a correction function $c$ forms a Pauli flow with some partial order $\prec$, then ${\prec_c} \sse {\prec}$, \ie any Pauli flow partial order must contain the induced partial order \cite[Lemma~2.11]{mitosekAlgebraic2024}.
Moreover, there exists a partial order $\prec$ such that $(c,\prec)$ is a Pauli flow if and only if $(c, \prec_c)$ is a Pauli flow \cite[Theorem~2.12]{mitosekAlgebraic2024}.
Therefore, in Sections~\ref{s:gflow} and~\ref{s:other-labels}, we will generally work with induced partial orders as that simplifies the proofs.
By the above results, this is without loss of generality.
One advantage of the induced partial order for a \emph{focused} correction function is that all Pauli measurements are initial, which can be seen e.g.\ from Observation~\ref{obs:NC-interpretation}.

While in this paper we will consider the conditions under which inserting planar measured vertices into one-way computations preserves the existence of flow, for Pauli-$Z$ measured vertices, this question has already been resolved: such vertices can be inserted with any set of neighbours.

\begin{proposition}[$Z$-insertion {\cite[Proposition~4.1]{mcelvanneyCompleteFlowPreservingRewrite2023}}]\label{prop:Z-insertion}
 Let $\Gamma = (G, I, O,\ld)$ be a labelled open graph with Pauli flow, where $G=(V,E)$ and let $W\sse V$ be some arbitrary subset of the vertices.
 Then $\Gamma' = (G', I, O, \ld')$ has Pauli flow, where $G'=(V',E')$ with $V'= V\cup\{z\}$, $E' = E \cup \{\{w,z\}\mid w\in W\}$, and $\ld'(v) = \ld(v)$ for $v\neq z$ and $\ld'(z)=Z$.
\end{proposition}

\subsection{ZX-calculus toolkit for MBQC}

Labelled open graphs provide a way to reason about the properties of a resource state that are necessary for robust determinism.
Yet we are interested in rewrites that preserve both the existence of flow and also the computation itself.
The latter is ensured by the ZX-calculus \cite{coeckeInteractingQuantumObservables2011}, a graphical language for reasoning about quantum processes.
We give a short introduction, more detail can be found e.g.\ in \cite{coeckePicturingQuantumProcesses2017,vandeweteringZXcalculusWorkingQuantum2020}.

\begin{definition}
    The ZX-calculus is the language generated by green $Z$-spiders, red $X$-spiders and yellow Hadamard nodes with the following interpretations:
    \begin{align*}
     \interp{\input{def/def_zspider.tikz}} &= \ketbra{0...0}{0...0} + e^{i\alpha} \ketbra{1...1}{1...1} \\
     \interp{\input{def/def_xspider.tikz}} &= \ketbra{+...+}{+...+} + e^{i\alpha}\ketbra{-...-}{-...-} \\
     \interp{\input{def/def_hspider.tikz}} &= \ketbra{0}{+} + \ketbra{1}{-}
    \end{align*}
    When green or red nodes have phase $\alpha \equiv 0 \bmod 2\pi$, we omit the phase label in the diagram.
    The components compose in parallel and in series like gates in a quantum circuit.
\end{definition}
A Hadamard node between two spiders can be thought of as a special sort of edge, represented as follows: $\input{def/def_hadamard-edge1.tikz} := \input{def/def_hadamard-edge0.tikz}$.
The ZX-calculus comes with an equational theory, which allows rewriting of diagrams in a way that preserves their interpretation.
There are multiple complete sets of rewrite rule \cite{backensZXcalculusCompleteStabilizer2014,jeandelCompleteAxiomatisationZXCalculus2018,vilmartNearOptimalAxiomatisationZXCalculus2018}, but we introduce only a subset of rules that is useful for our purposes, see Figure~\ref{fig:ZX-rules}.
Furthermore, we implicitly use the \textit{only connectivity matters} rule, which says that ZX-diagrams can in fact be treated as graphs without changing their interpretation.
This comes in handy to represent the open graphs underlying one way computations:
each qubit can be represented by a green node, connected to its neighbours via Hadamard edges, with a dangling edge on which we can plug the measurement effects, and potentially an additional dangling edge for
input qubits.
The measurement effects are represented by \input{measurement-effects/XY.tikz} for $XY$, \input{measurement-effects/YZ.tikz} for $YZ$, \input{measurement-effects/XZ.tikz} for $XZ$, \input{measurement-effects/X.tikz} for $X$,\input{measurement-effects/Y.tikz} for $Y$, and \input{measurement-effects/Z.tikz} for $Z$.
    
\begin{figure}
 \centering
  $\input{zx-rewrite-rules/0copy1.tikz}\hspace*{1em}\eq{}\hspace*{1em}\input{zx-rewrite-rules/0copy2.tikz}$
  \hfill
  $\input{zx-rewrite-rules/fusion1.tikz}\hspace*{1em}\eq{}\hspace*{1em}\input{zx-rewrite-rules/fusion2.tikz}$
  \hfill
  $\input{zx-rewrite-rules/color1.tikz}\hspace*{1em}\eq{}\hspace*{1em}\input{zx-rewrite-rules/color2.tikz}$

  $\input{zx-rewrite-rules/binZ1.tikz}\hspace*{1em}\eq{}\hspace*{1em}\input{zx-rewrite-rules/involh2.tikz}$
  \hfill
  $\input{zx-rewrite-rules/involh1.tikz}\hspace*{1em}\eq{}\hspace*{1em}\input{zx-rewrite-rules/involh2.tikz}$
  \hfill
  $\input{zx-rewrite-rules/copy1.tikz}\hspace*{1em}\eq{}\hspace*{1em}\input{zx-rewrite-rules/copy2.tikz}$
 \caption{The above rewrite rules for the ZX-calculus are sound.}
 \label{fig:ZX-rules}
\end{figure}
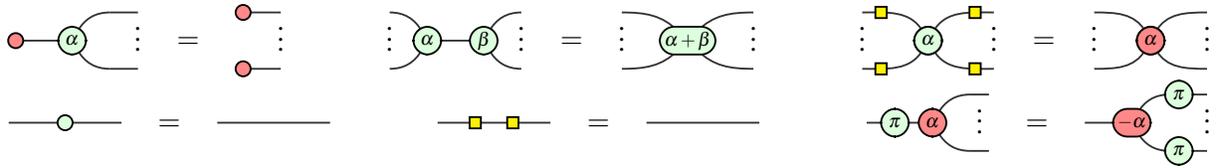

\begin{example}
 The open graph where $G$ is drawn on the left below, with $I = \{i_1,i_2\}$ and $O = \{o_1,o_2\}$ is represented the ZX-diagram on the right:

 \hfill\input{graph-state/graph.tikz} \hfill \input{graph-state/zx_graph.tikz}\hfill\;
\end{example}

\section{YZ-insertion into patterns with extended causal flow}\label{s:causal}

It might seem strange to discuss the insertion of $YZ$-measured vertices into measurement patterns with causal flow, since the traditional definition of causal flow allows only $XY$-measured vertices (cf.\ Definition~\ref{def:causal-flow}).
Yet the following generalisation to include $YZ$-measurements is straightforward\footnote{Indeed, while the authors came up with this definition during an internship in the spring/summer of 2023 \cite{perezMeasurementBased2023}, Calum Holker independently introduced an analogous generalisation around the same time \cite[Definition~27]{holkerCausal2023}.}.

    \begin{definition}\label{def_extcflow}
        Let $(G,I,O,\lambda)$ be a labelled open graph such that $\lambda(V) \subset \{XY,YZ\}$. $G$ has an \emph{extended causal flow} when there exists
        a map $c : \overline{O} \rightarrow \overline{I}$ and a strict partial order $\prec$ over $V$ such that $\forall u \in V$:
        \begin{multicols}{2}
            \begin{enumerate}
                \item[{\crtcrossreflabel{(C1)}[extcflow1]}] $u \prec c(u)$ or $u = c(u)$
                \item[{\crtcrossreflabel{(C2)}[extcflow2]}] $\lambda(u) = XY \Rightarrow c(u) \in N_G(u)$
                \item[{\crtcrossreflabel{(C3)}[extcflow3]}] $\lambda(u) = YZ \Rightarrow u = c(u)$
                \item[{\crtcrossreflabel{(C4)}[extcflow4]}] $\forall v \in N_G(c(u)), u \neq v \Rightarrow u \prec v$
            \end{enumerate}
        \end{multicols}
    \end{definition}

Measurement patterns with causal flow are very close to quantum circuits over the CNOT and single-qubit gate set.
Allowing $YZ$-measurements corresponds to allowing a new kind of gate: \emph{phase gadgets}, multi-qubit diagonal gates whose phase depends on the parity of the values of all qubits involved.
In other words, an $n$-qubit phase gadget with phase $\varphi$ acts on computational basis states as $\ket{x_1\ldots x_n} \mapsto e^{i\varphi (x_1\oplus\ldots\oplus x_n)}\ket{x_1\ldots x_n}$, \ie it uses XOR where a standard multi-controlled phase gate would use AND.

    \begin{definition}\label{def:YZ-insertion-cflow}
        Let $\Gamma = (G,I,O,\ld)$ be a \LOG.
        Define $\Gamma'=(G',I,O,\ld')$ to be the \LOG\ that results from inserting a new $YZ$-measured vertex $z$ with neighbourhood $S\sse V$, where $G' = (V\cup\{z\}, E\cup\{\{z,v\} \mid v \in S\})$ and $\ld':(\comp{O}\cup\{z\})\to\{XY, YZ\}$ is the extension of $\ld$ satisfying $\ld'(z)=YZ$ and $\ld'(v)=\ld(v)$ for all $v\in\comp{O}$.
    \end{definition}
    
    \begin{observation}\label{obs_XY_restriction}
        Put together, conditions \ref{extcflow3} and \ref{extcflow4} of Definition \ref{def_extcflow} forbid 
        $YZ$-measured vertices to be adjacent to one another in a \LOG\ with extended causal flow.
        We can therefore restrict ourselves to the case $\lambda(S) = \{XY\}$ without loss of generality.
    \end{observation}

    Analogous to gflow or Pauli flow \cite{duncanGraphtheoreticSimplificationQuantum2020,backensThereBackAgain2021,simmonsRelatingMeasurementPatterns2021}, $YZ$-measured vertices can be removed while preserving the existence of extended causal flow, as stated in the following lemma.
    Proofs are in Appendix~\ref{s:appendix-causal}.

    \begin{lemma}\label{lemma_remove_z}
        Let $ \Gamma = (G,I,O,\ld)$ be a labelled open graph, 
        and let $\Gamma'$ be the \LOG\ from Definition \ref{def:YZ-insertion-cflow}, with $\ld(S) \sse \{XY\}$.
        If $(c', \prec')$ is an extended causal flow for $\Gamma'$, then $(c'_{|V}, \prec'_{|V\times V})$ is an extended causal flow for $\Gamma$.
	\end{lemma}

    To determine when $YZ$-insertion into a \LOG\ $\Gamma$ preserves the existence of extended causal flow,
    we will first consider a fixed causal flow $(c, \prec)$ for $\Gamma$, and then generalise to any causal flow on $\Gamma$ using Lemma~\ref{lemma_remove_z}.
    Note that the uniqueness result for causal flow correction functions on \LOG{}s where $\abs{I}=\abs{O}$ \cite[p.~6]{deBeaudrapFinding2008} immediately generalises to extended causal flow since there is no choice for the correction set of a $YZ$-measured vertex.
    For the first step, we will need to extend the strict partial order $\prec$ to the newly inserted vertex $z$. The following lemma characterises when this is possible.

    \begin{lemma}\label{lemma_order}
		Let $\prec$ be a strict partial order on a set $V$, and suppose $z \not\in V$.
        Let $S_1, S_2 \subset V$, and let $\prec'$ be the transitive closure of $\prec''$, with ${\prec''} := {\prec} \cup (S_1 \times \{z\}) \cup (\{z\} \times S_2)$.
        Then $\prec'$ is a strict order if and only if $\forall v \in S_1,\forall w \in S_2. \neg(w \prec v \vee v = w)$.

        In other words, provided that no element of $S_2$ is equal to or precedes any element of $S_1$ in the original order $\prec$,
        we can extend the order $\prec$ to include $z$ such that $z$ succeeds all elements of $S_1$ and precedes all elements
        of $S_2$.
	\end{lemma}

    Due to the above lemma, it is natural to define an extended causal flow $(c', \prec')$ for $\Gamma'$ given an
    extended causal flow $(c, \prec)$ for $\Gamma$ as follows.
			
   \begin{definition}\label{def:c_prime}    
        Let $ \Gamma = (G,I,O,\ld)$ be a labelled open graph with extended causal flow $(c,\prec)$, 
        and let $\Gamma'$ be the \LOG\ from definition \ref{def:YZ-insertion-cflow}, with $\ld(S) \sse \{XY\}$. 
        Further define $c'$ to be the extension of $c$ to domain $(V\setminus O)\cup\{z\}$ that satisfies $c'(z)=z$, and let $\prec'$ be the transitive closure of ${\prec} \cup (\{z\}\times S) \cup (c^{-1}(S) \times \{z\})$.
    \end{definition}

    We can now state the main theorem of this section, characterising when the insertion of a $YZ$-measured vertex into a
    \LOG\ preserves the existence of extended causal flow.

	\begin{theorem}\label{characcausal}
		Let $\Gamma = (G,I,O,\lambda)$ be a \LOG. 
        Let $\Gamma'=(G',I,O,\ld')$ be the \LOG\ that results from inserting a new $YZ$-measured vertex $z$ with neighbourhood $S\sse V$ as in Definition~\ref{def:YZ-insertion-cflow},
        with $\ld(S) \sse \{XY\}$.
        Then $\Gamma'$ has extended causal flow if and only if there exists an extended causal flow $(c,\prec)$ for  $\Gamma$ such that $\forall v \in S, \forall w \in c^{-1}(S). \neg(v \prec w \vee v = w)$.

        In other words we can extend the causal flow to $\Gamma'$ if and only if no vertex of $S$ is equal to or precedes any vertex of $c^{-1}(S)$ in the original order $\prec$.
	\end{theorem}
    
\begin{example}\label{ex:causal-flow}
Consider $YZ$-insertion into the \LOG\ $\Gamma$ in \eqref{eq:example}, which has extended causal flow $(c,\prec)$ with $c$ given in \eqref{eq:example} and $i_1\prec i_2 \prec b \prec a$.
Let $S = \{a,b\}$. Then we have $c^{-1}(S) = \{i_1,i_2\}$, and $\forall v \in S, \forall w \in c^{-1}(S).\neg(v \prec w \vee v = w)$, so
by Theorem~\ref{characcausal}, the \LOG\ $\Gamma'$ of \eqref{eq:example} has extended causal flow where $c'(z)=z$, $c'(v)=c(v)$ otherwise, and $i_1 \prec' i_2 \prec' z \prec' b \prec' a$.
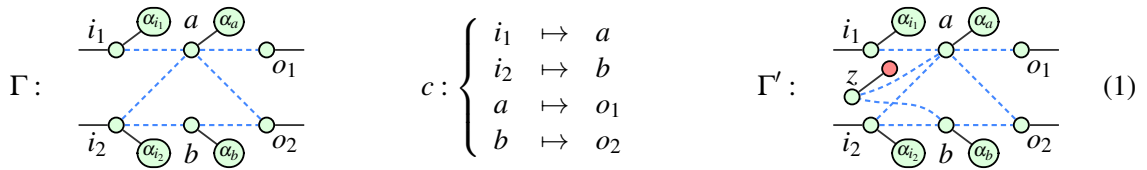
\begin{equation}\label{eq:example}
 \Gamma : \hspace*{1em} \input{graph-state/ex-caus1.tikz}
 \hspace*{4em}
 c :  \left\{
    \begin{array}{l c l }
        i_1 & \mapsto & a \\
        i_2 & \mapsto & b \\
        a & \mapsto & o_1 \\
        b & \mapsto & o_2 \\
    \end{array}
    \right.
 \hspace*{4em}
 \Gamma' : \hspace*{1em} \input{graph-state/ex-caus2.tikz}
\end{equation}
\end{example}

As a straightforward consequence of Theorem~\ref{characcausal}, if the $YZ$-insertion is to be flow preserving, then the number of neighbours of the new output can be no more than the number of outputs.

\begin{corollary}\label{cor:number-neighbours-outputs}
 Suppose inserting a $YZ$-measured vertex $z$ with neighbourhood $S$ as in Definition~\ref{def:YZ-insertion-cflow} into a \LOG\ $(G,I,O,\lambda)$ preserves the existence of extended causal flow.
 Then $\abs{S}\leq\abs{O}$.
\end{corollary}

We consider now the complexity of checking the conditions of Theorem~\ref{characcausal}.
For comparison, the complexity of the ($XY$-only) causal flow finding algorithm is linear in the number of edges in the \LOG\ \cite{mhallaFindingOptimalFlows2008}.
The algorithm generalises to extended causal flow by correcting $YZ$-measurements as soon as all their neighbours have been corrected.
Extended causal flows are unique for \LOG{}s where $\abs{I}=\abs{O}$ since causal flows are unique \cite{deBeaudrapFinding2008} and each $YZ$-measured vertex corrects itself.

The flow-preservation conditions involve only the neighbours of the new vertex and their preimages under the flow function.
On the other hand, the theorem gives an existence condition, so for \LOG{}s with $\abs{I}<\abs{O}$, different extended causal flows may need to be considered.
Additionally, local verification of the partial order conditions may not be possible since the flow finding algorithm does not return the induced partial order, but a layering of the vertices \cite{mhallaFindingOptimalFlows2008}.
Thus the flow-preservation conditions for $YZ$-insertion into extended causal flow may be of theoretical more than practical interest.

\section{YZ-insertion into patterns with gflow and Pauli flow}\label{s:gflow}

We will now consider the insertion of $YZ$-measured vertices into MBQC patterns with gflow and Pauli flow.
As gflow is a special case of Pauli flow, the two can be considered together: if a pattern has Pauli flow and all the measurements are planar, then it has gflow.
Inserting a $YZ$-measurement, which is planar, does not change that.
Thus, considering only Pauli flow is without loss of generality.
Furthermore, by Lemma~\ref{lem:focused-Pauli-flow}, it suffices to consider only focused Pauli flow.
Throughout this section we will assume that in any Pauli flow $(c,\prec)$, $\prec$ is the partial order induced by $c$ (cf.\ Definition~\ref{def:trl_c}), which is also without loss of generality.
We begin by proving an intermediate lemma that will be useful for the main theorem.
Proofs for this section are in Appendix~\ref{s:appendix-gflow}.

\begin{lemma}\label{lem:focusing-flow-demand}
 Let $\Gamma = (G,I,O,\ld)$ be a \LOG\ with flow-demand matrix $M$.
 Let $\someset\sse\comp{I}$ and denote the corresponding indicator vector (within $\comp{I}$) by $\v{a}$: \ie $a_u = 1 \Leftrightarrow u\in\someset$ and $a_u=0 \Leftrightarrow u\in\comp{I}\setminus\someset$.
 Then $\someset$ is focused over the set $S\sse\comp{O}$ if and only if $(M\v{a})_u = 0$ for all $u\in S$.
\end{lemma}

We are now ready to define the $YZ$-insertion and prove the associated flow-preservation theorem for gflow and Pauli flow.

\begin{definition}\label{def:YZ-insertion-gflow}
 Let $\Gamma = (G,I,O,\ld)$ be a \LOG.
 Define $\Gamma'=(G',I,O,\ld')$ to be the \LOG\ that results from inserting a new YZ-measured vertex $z$ with neighbourhood $S\sse V$, where
 $G' = (V\cup\{z\}, E\cup\{\{z,v\} \mid v \in S\})$,
 and $\ld':(\comp{O}\cup\{z\})\to\{X,Y,Z,XY,XZ,YZ\}$ is the extension of $\ld$ satisfying $\ld'(z)=YZ$ and $\ld'(v)=\ld(v)$ for all $v\in\comp{O}$.
\end{definition}

In the following, (unprimed) $\comp{I}$, $\comp{O}$, $\planar$ \etc will always refer to the vertices of the original \LOG\ $\Gamma$ only.
The $YZ$-insertion with a given set of neighbours is flow-preserving only if certain properties are satisfied: there must be a suitable focused correction set for the newly-inserted vertex, and the relation induced by the new correction function must be a strict partial order.
The first two conditions of the theorem ensure the existence of the focused correction set and the third condition captures the requirements for the partial order.

\begin{theorem}\label{thm:YZ-into-Pauli}
 Let $\Gamma = (G,I,O,\ld)$ be a \LOG\ and let $\Gamma'$ be the \LOG\ that results from inserting a new $YZ$-measured vertex $z$ with neighbourhood $S\sse V$ as in Definition~\ref{def:YZ-insertion-gflow}.
 Then $\Gamma'$ has Pauli flow if and only if there exists a Pauli flow $(c,\prec)$ on $\Gamma$ and a set $K\sse\comp{I}$ such that:
 \begin{enumerate}
  \item\label{it:K-focused} $K$ is focused over $\comp{O}\setminus(S\cap\Xlike)$, and moreover $\comp{O}\setminus(S\cap\Xlike)$ is the largest set over which $K$ is focused.
  \item\label{it:K-S-even} $\abs{S\cap K} \equiv 0 \bmod 2$.
  \item\label{it:order} If $\Wpred := \{w\in\comp{O}: \abs{c(w)\cap S}\equiv 1 \bmod 2\}$ and $\Vsucc := \planar\cap (K \cup (\odds{G}{K} \symd S))$, then all $w\in\Wpred$ and all $v\in\Vsucc$ satisfy $\neg(v\prec w \vee v = w)$;
  \ie no vertex of $\Vsucc$ is equal to or precedes any vertex of $\Wpred$ for the original order $\prec$.
 \end{enumerate}
 If the properties hold, the focused Pauli flow on $\Gamma'$is given by the extension $c'$ of $c$ to domain $\comp{O}\cup\{z\}$ satisfying $c'(z) = K\cup\{z\}$,
 and the transitive closure of ${\prec} \cup \{(w,z)\mid w\in\Wpred\} \cup \{(z,v)\mid v\in\Vsucc \}$.
\end{theorem}

\begin{example}
Consider the \LOG\ $\Gamma$ from Example~\ref{ex:causal-flow}, this has focused Pauli flow $(c,\prec)$ (which is actually a gflow) with
$c(i_1) = \{a,b\}$, $c(i_2) = \{b\}$, $c(a) = \{o_1\}$, $c(b) = \{o_1, o_2\}$ and $i_1,i_2 \prec a,b$.
Let $S = \{i_1,a,b\}$ and $K = \{a,b,o_2\}$. First, $\comp{O}\setminus(S\cap\Xlike) = \{ i_2\}$. As $\odds{G}{K} = \{i_1,a,b\}$,
one can check that indeed $\{i_2\}$ is the largest set on which $K$ is focused.
Furthermore, $|K\cap S| = |\{1,2\}| \equiv 0 \bmod 2$. Finally, $\Wpred = \{i_2\}$ and $\Vsucc = \{a,b\}$, 
thus for all $w\in\Wpred$ and $v\in\Vsucc$ we have $\neg(v\prec w \vee v = w)$.
Thus by Theorem~\ref{thm:YZ-into-Pauli}, the following \LOG\ has focused gflow $(c', \prec')$:
    \[\input{graph-state/ex-caus3.tikz}\hspace*{2em} \text{with } c'(z) = \{z, 1, 2, o_2\}, \text{ and }
    \left\{\begin{array}{ c l c l c l c}
    i_2 & \prec' & z & \prec' & a,b & \prec' & o_1, o_2\\
     & & i_1 & \prec' & a,b & \prec' & o_1, o_2.
    \end{array}
    \right.
           \]
\end{example}

It has long been known that the `$Z$-deletion' rule is sound also for $YZ$-measurements: they can always be deleted from a \LOG\ without breaking the existence of flow \cite{duncanGraphtheoreticSimplificationQuantum2020,backensThereBackAgain2021,simmonsRelatingMeasurementPatterns2021}.
This is captured by our theorem via the following corollary, which follows from the second part of the proof above.

\begin{corollary}\label{cor:YZ-deletion}
 The three conditions in Theorem~\ref{thm:YZ-into-Pauli} are satisfied for any focused Pauli flow on $\Gamma'$, therefore $YZ$-deletion is always flow-preserving.
\end{corollary}

In the other direction, for $YZ$-insertion, it must be checked each time whether the three conditions are satisfied by the chosen set of neighbours $S$ of the new vertex and the proposed correction set (minus the vertex itself).
Yet there are some cases where the conditions are always satisfied: for example, $Z$-measured neighbours `come for free' in the sense that they will not affect whether or not the $YZ$-insertion is flow-preserving: this can be seen by inserting the $YZ$-measurement $u$ with non-$Z$ neighbours only, and then in turn deleting each desired $Z$-measured neighbour and re-inserting it with the extra edge to $u$.
As $Z$-deletions and -insertions always preserve Pauli flow (cf.\ Proposition~\ref{prop:Z-insertion}), the entire process is flow-preserving if and only if the original $YZ$-insertion preserves the existence of flow.
Any proof that the old \LOG\ has Pauli flow if the new \LOG\ has it follows immediately from Corollary~\ref{cor:YZ-deletion}, so we only have to show that $\Gamma$ having Pauli flow implies that $\Gamma'$ has Pauli flow.
Two additional corollaries handling all-input and all-output neighbours, respectively, are in Appendix~\ref{s:appendix-gflow}.

\begin{corollary}\label{cor:YZ-insertion-single-neighbour}
 Let $\Gamma = (G,I,O,\ld)$ be a \LOG\ with focused Pauli flow $(c,\prec)$ and let $x\in V$ be s.t.\ $\ld(x)=XY$.
 Define $\Gamma'$ to be the \LOG\ that results from inserting a new $YZ$-measured vertex $z$ with neighbourhood $S=\{x\}$ as in Definition~\ref{def:YZ-insertion-gflow}.
 Then $\Gamma'$ has focused Pauli flow $(c',\prec')$ where $c'$ is the extension of $c$ satisfying $c'(z) = c(x)\cup\{z\}$, and ${\prec'} = {\prec} \cup \{(w,z)\mid w\prec x\} \cup \{(z,v)\mid x\prec v\}$.
\end{corollary}

\begin{remark}
 Corollary~\ref{cor:YZ-insertion-single-neighbour} shows one way of getting a Pauli flow on $\Gamma'$ for any focused Pauli flow on the original \LOG\ $\Gamma$.
 Yet not all focused Pauli flows on $\Gamma'$ necessarily arise in this way: consider, for example, the line graph $G = (\{x,o_1,o_2\}, \{\{o_1,x\},\{x,o_2\}\})$ and the resulting \LOG\ $(G, \emptyset, \{o_1,o_2\}, x\mapsto XY)$.
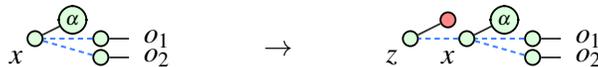
\[
\input{rmk-4-9/1.tikz}\hspace*{3em}\DtoD{}\hspace*{3em}\input{rmk-4-9/2.tikz}
\]
 This has two focused Pauli flows, with $c(x)=\{o_1\}$ or $c(x)=\{o_2\}$.
 After inserting a new $YZ$-type measurement $z$ with sole neighbour $x$, there are four focused Pauli flows on the resulting \LOG: either of the correction sets $\{z,o_1\}$ and $\{z,o_2\}$ works for $z$, independent of the choice of correction set for $x$.
\end{remark}

If a \LOG\ represents a unitary operator -- \ie the number of outputs equals the number of inputs -- then focused Pauli flow, if it exists, is unique (since the flow-demand matrix is square and thus has at most one right inverse).
In this case, the possible correction set of the new vertex is entirely determined by the choice of $S$.

\begin{observation}\label{obs:I-O-unique-correction}
 Suppose $\Gamma = (G,I,O,\ld)$ is a \LOG\ such that $\abs{I}=\abs{O}$, and it has Pauli flow.
 Let $S\sse V$ and define $\Gamma'$ to be the \LOG\ that results from inserting a new $YZ$-measured vertex $z$ with neighbourhood $S$ as in Definition~\ref{def:YZ-insertion-gflow}.
 Write $M$ for the flow-demand matrix of $\Gamma$, $\v{x}$ for the indicator vector of $S\cap\Xlike$ in $\comp{O}$, and $\v{k}$ for the indicator vector of an (as yet unknown) set $K\sse\comp{I}$ such that $K\cup\{z\}$ is a focused correction set for $z$.
 Condition~\ref{it:K-focused} is equivalent to $M\v{k} = \v{x}$ by Lemma~\ref{lem:focusing-flow-demand}.

 Since $\Gamma$ has Pauli flow, $M$ has a right inverse $C$.
 Now, $\abs{I}=\abs{O}$ implies that $M$ is square, therefore $C$ is in fact the unique two-sided inverse for $M$.
 Hence $M\v{k}=\v{x}$ if and only if $\v{k} = C\v{x}$.
 The set corresponding to $C\v{x}$ satisfies condition~\ref{it:K-focused} by construction, yet it need not satisfy the other two conditions of Theorem~\ref{thm:YZ-into-Pauli}.
\end{observation}

In the general case, it is difficult to analyse the complexity of checking the conditions in Theorem~\ref{thm:YZ-into-Pauli} for a given set of neighbours $S$ as the theorem requires only that there exists a Pauli flow which satisfies them.
Yet if the \LOG\ satisfies $\abs{I}=\abs{O}$ as in Observation~\ref{obs:I-O-unique-correction}, then focused Pauli flow is unique and that difficulty disappears.
Suppose this unique Pauli flow on $\Gamma$ is known in the form of the correction matrix $C$ and the product $N C$, where $N$ is the order-demand matrix.
We will assume $M$ and $N$ are known too (otherwise they can be constructed in time $\bigO(n^2)$).
From $S$, we can construct in linear time the indicator vectors $\v{x}$, $\v{n}$ and $\v{z}$, and then compute in $\bigO(n^2)$ the indicator vector of the correction set $\v{k} = C\v{x}$.
Then, again in $\bigO(n^2)$, we can compute the additional row and column for the product $N_{\Gamma'} C_{\Gamma'}$.
Checking whether the resulting matrix represents a directed acyclic graph is also in $\bigO(n^2)$ (see \cite[Section 20.4]{cormenIntroductionAlgorithmsFourth2022}).
Thus the entire process of checking the conditions and updating the flow proceeds in $\bigO(n^2)$, which is more efficient than finding a flow from scratch (the latter has the complexity of matrix multiplication, \ie $\bigO(n^{\log_2 7})$ using Strassen's algorithm).
Thus, if the number of vertices to be inserted is less than $\bigO(n^{\log_2 7 - 2})$, it is more efficient to do this step by step and update the known flow, rather than inserting all the vertices and re-running the flow-finding algorithm.

Another consequence of Observation~\ref{obs:I-O-unique-correction} is that, if $\abs{S\cap K}\equiv 1\bmod 2$ for the unique $K$ that satisfies condition~\ref{it:K-focused}, then it is not possible to insert a $YZ$-measured vertex with neighbourhood $S$, as $u$ would be in the odd neighbourhood of its own correction set, contradicting \ref{P6}.
Yet in that case, the right-hand side of \ref{P5} would be satisfied for a vertex with that neighbourhood: so it may be possible to insert an $XZ$-measured vertex with neighbourhood $S$ instead, assuming the partial order conditions hold.
Thus, in the next section, we will consider inserting vertices with measurement labels other than $YZ$ (or $Z$).

\section{Inserting vertices with other measurement labels}
\label{s:other-labels}

$XZ$-measurements behave similarly to $YZ$-measurements: the focusing conditions are the same for both; equivalently, they give rise to the same types of rows in the flow-demand matrix.
The rows of the order-demand matrix corresponding to an $XZ$- or a $YZ$-measured vertex are also nearly the same, the only difference is at the intersection of the column corresponding to the same vertex: a `1' in that position encodes that $XZ$-measurements are in the odd neighbourhood of their own correction set whereas a `0' in that position encodes that $YZ$-measurements are not.
The flow preservation theorem for $YZ$ insertions can thus straightforwardly be adapted to $XZ$-measurements instead.
Proofs for this section are in Appendix~\ref{s:appendix-other-labels}.

\begin{definition}\label{def:XZ-insertion-gflow}
 Let $\Gamma = (G,I,O,\ld)$ be a \LOG.
 Define $\Gamma'=(G',I,O,\ld')$ to be the \LOG\ that results from inserting a new XZ-measured vertex $u$ with neighbourhood $S\sse V$, where
 $G' = (V\cup\{z\}, E\cup\{\{z,v\} \mid v \in S\})$,
 and $\ld':(\comp{O}\cup\{z\})\to\{X,Y,Z,XY,XZ,YZ\}$ is the extension of $\ld$ satisfying $\ld'(z)=XZ$ and $\ld'(v)=\ld(v)$ for all $v\in\comp{O}$.
\end{definition}

\begin{theorem}\label{thm:XZ-into-Pauli}
 Let $\Gamma = (G,I,O,\ld)$ be a \LOG\ and let $\Gamma'$ be the \LOG\ that results from inserting a new $XZ$-measured vertex $z$ with neighbourhood $S\sse V$ as in Definition~\ref{def:XZ-insertion-gflow}.
 Then $\Gamma'$ has Pauli flow if and only if there exists a Pauli flow $(c,\prec)$ on $\Gamma$ and a set $K\sse\comp{I}$ such that:
 \begin{enumerate}
  \item $K$ is focused over $\comp{O}\setminus(S\cap\Xlike)$, and moreover $\comp{O}\setminus(S\cap\Xlike)$ is the largest set over which $K$ is focused.
  \item $\abs{S\cap K} \equiv 1 \bmod 2$.
  \item Let $\Wpred := \{w\in\comp{O} : \abs{c(w)\cap S}\equiv 1 \bmod 2\}$ and $\Vsucc := \planar\cap (K \cup (\odds{G}{K} \symd S))$, then all $w\in\Wpred$ and all $v\in\Vsucc$ satisfy $\neg(v\prec w \vee v = w)$;
  \ie no vertex of $\Vsucc$ is equal to or precedes any vertex of $\Wpred$ for the original order $\prec$.
 \end{enumerate}
 If the properties hold, the focused Pauli flow on $\Gamma'$is given by the extension $c'$ of $c$ to domain $\comp{O}\cup\{z\}$ satisfying $c'(z) = K\cup\{z\}$,
 and the transitive closure of ${\prec} \cup \{(w,z)\mid w\in\Wpred\} \cup \{(z,v)\mid v\in\Vsucc \}$.
\end{theorem}

\begin{proof}
 The proof is analogous to that of Theorem~\ref{thm:YZ-into-Pauli} with the only difference being that $N_{\Gamma'} = \smm{N & \v{z} \\ \v{n} & 1}$, which corresponds to the changed requirement that $\abs{S\cap K} \equiv 1 \bmod 2$.
\end{proof}

We have now shown how to insert $YZ$ or $XZ$-measured vertices while preserving flow, and it was already known that $Z$-insertion preserves the existence Pauli flow (cf.\ Proposition~\ref{prop:Z-insertion}).
Now, inserting an $X, Y$ or $XY$ measurements into an MBQC pattern must change the interpretation, so it does not make sense to look for insertion rules for those measurement types.
Nevertheless, we can use the existing insertion rules in and then apply local complementations or pivots that change the measurement label of the new vertex.
This effectively inserts a measurement of a different type into the pattern, at the cost of introducing additional changes elsewhere in the \LOG.

Indeed, with just one appeal to standard ZX-calculus rewrite rules to handle phase labels and a new lemma that explicitly shows how flow changes under the pivot operation, we can use $YZ$-insertions to prove the `vertex splitting' or `neighbour unfusion' rule that has previously been used in various forms in the literature about rewriting MBQC and ZX-calculus diagrams \cite{staudacherReducing2QuBitGate2023,mcelvanneyFlowpreservingZXcalculusRewrite2023,holkerCausal2023}.
This connection was previously pointed out in \cite[p.~215]{mcelvanneyFlowpreservingZXcalculusRewrite2023}.
The vertex splitting rule always preserves Pauli flow if the new phase-0 vertex is taken to be an $X$-measurement \cite[Corollary 6.1]{mcelvanneyFlowpreservingZXcalculusRewrite2023}, yet if that vertex is taken to be an $XY$-measurement, flow is not always preserved \cite[Section~5]{staudacherReducing2QuBitGate2023}.

\begin{lemma}\label{lem:pivot-Pauli-flow}
 Suppose the \LOG\ $(G,I,O,\ld)$ has Pauli flow $(c,\prec)$.
 Let $u,v\in\comp{I}\cap\comp{O}$ such that $\{u,v\}\in E$.
 Define $\ld':\comp{O}\to\{XY,XZ,YZ,X,Y,Z\}$ and $c':\comp{O}\to\pow{\comp{I}}$ as
 \[
  \ld'(w) := \begin{cases}
              Z &\text{if } w\in\{u,v\}\wedge\ld(w) = X \\
              X &\text{if } w\in\{u,v\}\wedge\ld(w) = Z \\
              YZ &\text{if } w\in\{u,v\}\wedge\ld(w) = XY \\
              XY &\text{if } w\in\{u,v\}\wedge\ld(w) = YZ \\
              \ld(w) &\text{otherwise.}
             \end{cases}
  \qquad
  c'(w) := \begin{cases}
            c(w) &\text{if } u,v\notin\codds{G}{c(w)} \\
            c(w)\symd\{u\} &\text{if } u\in\codds{G}{c(w)} \not\ni v \\
            c(w)\symd\{v\} &\text{if } u\notin\codds{G}{c(w)} \ni v \\
            c(w)\symd\{u,v\} &\text{if } u,v\in\codds{G}{c(w)},
           \end{cases}
 \]
 then $(G\wedge uv, I, O, \ld')$ has Pauli flow $(c',\prec)$.
 Moreover, for any $w\in\comp{O}$, we have $\codds{G\wedge uv}{c'(w)} = \codds{G}{c(w)}$.
 If $(c,\prec)$ is focused and $\ld(u),\ld(v)\neq XZ$, then $(c',\prec)$ is also focused.
 Finally, if $(c,\prec)$ is a gflow, then $(c',\prec)$ is also a gflow.
\end{lemma}

Interestingly, this lemma is rather simpler than equivalent ones for local complementations \cite{backensThereBackAgain2021,simmonsRelatingMeasurementPatterns2021}, despite pivots arising from the former.
We are now ready to consider the vertex splitting operation.

\begin{definition}\label{def:vertex-splitting}
 Let $\Gamma=(G,I,O,\ld)$ be a \LOG\, let $x\in\comp{O}\cap\comp{I}$ be such that $\ld(x)=XY$, and let $W\sse V\setminus\{x\}$.
 The \LOG\ $\Gamma'=(G',I,O,\ld')$ that results from splitting $x$ over the set $W$ is defined as follows:
 \begin{itemize}
  \item Let $V' = V \cup \{x',x''\}$, \ie we add two new vertices $x',x''$ to the graph.
  \item Let $E' = E \symd \{\{x,x'\},\{x',x''\}\} \symd (\{x,x''\}\times W$), \ie we connect $x$ to $x'$ and $x'$ to $x''$, toggle the edges between $x$ and elements of $W$, and add edges between $x''$ and elements of $W$.
  \item Let $\ld'$ be the extension of $\ld$ to domain $V'$ which satisfies $\ld'(x')=\ld'(x'')=XY$.
 \end{itemize}
\end{definition}

This definition is slightly more general than those appearing in the literature, which usually restrict $W$ to be a subset of $N_G(x)$ or even just a single element of $N_G(x)$ \cite{staudacherReducing2QuBitGate2023}.

We will explicitly state the condition of the vertex splitting rule for the case $\abs{I}=\abs{O}$ only.
This is both because the motivating neighbour unfusion rule is mainly used within diagrams arising from circuits -- \ie where the numbers of inputs and outputs are equal -- and because the flow-preservation conditions get even more complicated in the general case where we would need to consider two separate unknown sets of vertices, one for the correction set of $x'$ and the other for $x''$.

\begin{theorem}\label{thm:vertex-splitting}
 Let $\Gamma=(G,I,O,\ld)$ be a \LOG\ such that $\abs{I}=\abs{O}$, let $x\in\comp{O}\cap\comp{I}$ be such that $\ld(x)=XY$, and let $W\sse V\setminus\{x\}$.
 Suppose $\Gamma'$ is the \LOG\ that results from splitting $x$ over the set $W$ as in Definition~\ref{def:vertex-splitting}.
 Then $\Gamma'$ has Pauli flow if and only if there exists a focused Pauli flow $(c,\prec)$ on $\Gamma$ such that:
 \begin{enumerate}
  \item[{\crtcrossreflabel{(V1)}[VS1]}] If $K$ is the unique set that is both focused over $\comp{O}\setminus (W\cap\Xlike)$ and has the property that $\comp{O}\setminus (W\cap\Xlike)$ is the largest set over which $K$ is focused, then $\abs{W\cap K} \equiv 0 \bmod 2$.
  \item[{\crtcrossreflabel{(V2)}[VS2]}] If $\Wpred := \{w\in\comp{O} : \abs{c(w)\cap W}\equiv 1 \bmod 2\}$ and $\Vsucc := \mathcal{L}\cap (K\cup(\odds{G}{K}\symd W))$, then all $w\in\Wpred$ and all $v\in\Vsucc$ satisfy $\neg(v\prec w \vee v = w)$;
  \ie no vertex of $\Vsucc$ is equal to or precedes any vertex of $\Wpred$ for the original order $\prec$.
  \item[{\crtcrossreflabel{(V3)}[VS3]}] For all $u\in V$, we have $x\prec u \implies u\notin\Wpred$ and $u\prec x \implies u\notin\Vsucc$;
  \ie vertices in $\Wpred$ do not succeed $x$ and vertices in $\Vsucc$ do not precede $x$ for the original order $\prec$.
  \item[{\crtcrossreflabel{(V4)}[VS4]}] $x\in K$ if and only if $\abs{c(x)\cap W}\equiv 0\bmod 2$.
 \end{enumerate}
 If all measurements in $\Gamma$ are planar, then so are all measurements in $\Gamma'$, thus the same conditions imply preservation of gflow.
\end{theorem}

Regarding the restricted version of neighbour unfusion, where $W$ consists of exactly one neighbour of $x$ (as used by Staudacher et al.~\cite{staudacherReducing2QuBitGate2023}), we find the following corollary.

\begin{corollary}\label{cor:neighbour-unfusion}
 Let $\Gamma=(G,I,O,\ld)$ be a \LOG\ such that $\abs{I}=\abs{O}$, let $a\in\comp{O}\cap\comp{I}$ be such that $\ld(a)=XY$, and let $b\in N_G(a)$ be such that $\ld(b)=XY$.
 Suppose $\Gamma'$ is the \LOG\ that results from splitting $a$ over the set $W:=\{b\}$ as in Definition~\ref{def:vertex-splitting}.
 Then $\Gamma'$ has Pauli flow if and only if there exists a focused Pauli flow $(c,\prec)$ on $\Gamma$ such that:
 \begin{itemize}
  \item Either $a\in c(b)$ or $b\in c(a)$; \ie one of the two vertices is in the correction set of the other.
  \item For all $u\in V$, we have $a\prec u \implies \neg(u\prec b)$ and $u\prec a \implies \neg(b\prec u)$; \ie there is no third vertex in
  between $a$ and $b$ in the original order $\prec$.
 \end{itemize}
 If all measurements in $\Gamma$ are planar, then so are all measurements in $\Gamma'$, thus the same conditions imply preservation of gflow.
\end{corollary}

This shows that the necessary and sufficient flow-preservation criterion derived by McElvanney for neighbour unfusion in the setting of gflow~\cite{McElvanneyThesis} generalises to Pauli flow as well.

\begin{remark}
 The conditions of Corollary~\ref{cor:neighbour-unfusion} are symmetric under interchange of $a$ and $b$.
 This reflects the fact that, when both new vertices are treated as having planar measurements (as we have done), then splitting $a$ over the set $\{b\}$ results in the same labelled open graph as splitting $b$ over the set $\{a\}$.

 If only the first insertion of Lemma~\ref{lem:vertex-splitting-steps} is treated as a planar measurement and the second insertion is treated as a Pauli measurement, then the existence of Pauli flow is preserved for all choices of~$W$ \cite[Proposition~5.1]{mcelvanneyFlowpreservingZXcalculusRewrite2023}.
 On the other hand, gflow is not preserved since $\Gamma'$ contains at least one non-planar measurement, and the symmetry between $a$ and $b$ is broken.
\end{remark}

\section{Conclusions}
\label{s:conclusions}

We have derived conditions under which planar-measured qubits can be inserted into one-way computations without breaking different types of flow.
This is a step towards a better theoretical understanding of flow-preserving rewriting for MBQC and for the ZX-calculus.
In applications such as the reduction of two-qubit gate counts, it is desirable to remain within the framework of gflow instead of inserting Pauli-measured vertices that have a different extraction procedure \cite{staudacherReducing2QuBitGate2023}: we have characterised when this is possible.

By choosing suitable measurement angles, the insertion of planar measurements into a one-way computation or a ZX-diagram can be performed without changing the semantics, i.e.\ without changing the linear operator that is implemented; this is desired for example in the above two-qubit gate count reduction procedure.
Yet by allowing measurement angles to vary, flow-preserving insertion of one (or more) planar measurements can also increase the range of linear operators that are implementable using a given MBQC pattern or ZX-diagram.
This has applications in the generation of ans\"atze for variational eigensolvers within measurement-based \cite{fergusonMeasurement-based2021} or ZX-based settings \cite{ewenApplication2024}.

In the common case of working with unitary computations or ZX-diagrams (where the number of inputs is equal to the number of outputs), checking the flow-preservation conditions and then updating the gflow or Pauli flow is more efficient than finding a flow on the new \LOG\ from scratch.
The same is not guaranteed if there are more outputs than inputs: in that situation, there are generally different flows on the same \LOG, and the insertion may be flow-preserving for some and not others.
The question of whether there are efficient ways to modify a known flow to make it compatible with an insertion for a given neighbourhood is left to future work.

\subsection*{Acknowledgments}

The authors would like to thank Tommy McElvanney, Piotr Mitosek, and Korbinian Staudacher for interesting and useful conversations during the course of the internship on which this paper is based.

\bibliographystyle{eptcs}
\bibliography{refs}

\appendix
\section{Proofs for Section~\ref{s:causal}}
\label{s:appendix-causal}

    \begin{proof}[Proof of Lemma~\ref{lemma_remove_z}]
		Let $u \in V$. Assume for a contradiction that $c'(u) = z$.
        As $u \neq z$, we have by condition~\ref{extcflow3} that $\lambda(u) = XY$, which implies by condition~\ref{extcflow2} that $z \in N_G(u)$.
        Then, condition~\ref{extcflow4} gives $z = c'(u) \prec' u$, which contradicts condition~\ref{extcflow1}.

        Therefore, $c'^{-1}(z) = \{z\}$, and $c'_{|V}$ is well defined.

        We have that $c'_{|V}(u) = c'(u) \in V$. Then as $u \prec' c'(u)$ or $u = c'(u)$, we have $u \prec'_{|V\times V} c'_{|V}(u)$ or $u = c'_{|V}(u)$.
        Condition~\ref{extcflow1} is therefore satisfied.
        Similarly, as the connectivity, labelling and correction are the all induced by $\Gamma'$, and as nothing changes on $V$,
        conditions~\ref{extcflow2},~\ref{extcflow3} and~\ref{extcflow4} are satisfied.
	\end{proof}

    \begin{proof}[Proof of Lemma~\ref{lemma_order}]
		Let's first show the direct implication by contraposition.

			Let $v,w \in S_1 \times S_2$ such that $ w \prec v \vee v = w$. We have $v\in S_1 \Rightarrow v \prec' z$ and  $w\in S_2 \Rightarrow z \prec' w$ which gives $v \prec' w$.
			As we have $w \prec v \vee v = w$, in both cases we get $w \prec' w$, and therefore $\prec'$ is not strict.

		Now, for the converse implication, suppose that for all $v\in S_1, w\in S_2. \neg(w \prec v \vee v = w)$. Let us show
        that $\prec'$ is a strict order.

				We directly have that $\prec'$  is transitive as the transitive closure of
				$\prec''$. Assume for a contradiction that it is not strict.

				Let $v \in V\cup \{z\}$ such that $v\prec'v$.
				Then as $\prec'$ is the transitive closure of $\prec''$, we have \[\exists n >1: \exists (v_1,\ldots,v_n) \in V: v = v_1 \prec'' v_2 \prec'' \ldots \prec'' v_n = v\]
				Denote by $[n]$ the set of all integers between 1 and $n$ (inclusive).

				If $n = 2$, we have $v \prec'' v$, which can only be true if $v \prec v$ due to the definition of $\prec''$. This contradicts the strictness of $\prec$.

				If $n>2$, the same argument can be used if we have $v_1 \prec v_2 \prec \ldots \prec v_n$.

                If not, the construction of $\prec''$ gives
					that \[\exists i \in [n] : v_i = z\]

					If $v \neq z$, denote \[i_{min} = \min\{i \in [n-1]\setminus\{1\} \mid v_i = z \},\hspace*{1em} i_{max} = \max\{i \in [n-1]\setminus\{1\} \mid v_i = z \}.\]
                        Then if we denote ${\preceq} = {\prec} \cup {=}$, we have that $v \preceq v_{i_{min}-1} $ and $v_{i_{max}+1} \preceq v$.
						This gives $v_{i_{max}+1} \preceq v  \preceq v_{i_{min}-1}$ with $v_{i_{min}-1} \in S_1$, and $v_{i_{max}+1} \in S_2$
						giving that either $v_{i_{max}+1} \prec v_{i_{min}-1}$ or $v_{i_{max}+1} = v_{i_{min}-1}$, which contradicts the order hypothesis in both cases.

					Now if $v = z$, we have that $v_2 \neq z$. We then have $v_2 \prec' z \prec' v_2$. We can then use the previous argument with $v \leftarrow v_2$ to get a contradiction.
	\end{proof}

	\begin{proof}[Proof of Theorem~\ref{characcausal}]
		Let us begin with the direct implication.

			Let's suppose that $G'$ has an extended causal flow $(c', \prec')$. Then let's suppose for a contradiction that
			\[\exists v \in S : \exists w \in c'^{-1}(S): v \prec w\]
		    As $v \in S$ we have $v \in N_{G'}(z) = N_{G'}(c'(z))$ and therefore $z \prec' v$.
            Besides, $w \in c'^{-1}(S) $ gives that $z \in N_{G'}(c'(w))$, and therefore $w \prec' z$.
            This gives $z \prec' z$ which contradicts the strictness of $\prec'$.
            A similar argument holds if $\exists v \in S : \exists w \in c'^{-1}(S): v = w$.

            Therefore, $(c', \prec')$ is such that $\forall v \in S, \forall w \in c'^{-1}(S), \neg(v \prec' w)$. Using Lemma~\ref{lemma_remove_z} to remove $z$, we get
			an extended causal flow $(c,\prec)$ for  $(G,I,O,\lambda)$ with the same property, implying $\forall v \in S, \forall w \in c^{-1}(S), \neg(v \prec w \vee v = w)$.
			\\

        Let us now show the converse implication.
            Let $(c,\prec)$ be an extended causal flow for $(G,I,O,\lambda)$ such that $\forall v \in S, \forall w \in f^{-1}(S), \neg(v \prec w \vee v = w)$.
            Let $(c', \prec')$ be as in Definition~\ref{def:c_prime}.
            Using Lemma~\ref{lemma_order} with $S_1 \leftarrow c^{-1}(S)$ and $S_2 \leftarrow S$ we get that $\prec'$ is a strict order.

            Let us now show that $(c', \prec')$ is an extended causal flow.
            One can easily check that $(c',\prec')$ satisfies conditions \ref{extcflow1}, \ref{extcflow2}, \ref{extcflow3} of Definition~\ref{def_extcflow} by construction.
            Let's elaborate on why it satisfies condition \ref{extcflow4}.

            Let $v \in V\cup\{z\}$ and $w \in N_{G'}(c(v))$ such that $ v\neq w$. Let's show that $v \prec' w$:

                If $v=z$, then $N_{G'}(c'(v)) = S$, thus $w\in S$, and by construction $v \prec' w$.

                Otherwise, $v\neq z$, and then $c'(v) = c(v) \neq z$.
                Recall that $E' =  E\cup\{\{z,v\} \mid v \in S\}$.
                Then as we have $w \in N_{G'}(c(v))$, either $w \in N_G(c(v))$ or $\{w,c(v)\} \in \{\{z,v\} \mid v \in S\}$.

                If $w \in N_G(c(v))$, as $v\neq w$ we get that $v\prec w$ from condition \ref{extcflow4} of $(c,\prec)$ being an extended causal flow.
                By construction we get $v \prec' w$.

                Otherwise, $w=z$ and $c(v) \in S$, as the case $c(v) = z$ would imply $v = z$.
                Then $v \in c^{-1}(S)$. Therefore by construction we have $v \prec' w$.

        Therefore $v \prec' w$.
           Therefore $\Gamma'$ has extended causal flow.

	\end{proof}

\begin{proof}[Proof of Corollary~\ref{cor:number-neighbours-outputs}]
 If all measurements are $XY$, the causal flow function $c$ defines a path covering of the \LOG, in which each path is given by a finite sequence $v, c(v), c(c(v)),\ldots$ for some $v\in V$, with the final element of the sequence being an output.
 Thus the number of paths is equal to the number of outputs \cite[Section~III]{deBeaudrapFinding2008}.
 By Lemma~\ref{lemma_remove_z}, this generalises to extended causal flow, where the paths cover only the $XY$-measured vertices.

 Suppose for a contradiction that $\abs{S} > \abs{O}$, then by the pigeonhole principle there exist two distinct vertices $s_1,s_2\in S$ which belong to the same path.
 Without loss of generality assume $s_2$ is the one that is closer to the output, \ie $s_2 = c^j(s_1)$ for some positive integer $j$.
 But then $c^{j-1}(s_1) = c^{-1}(s_2)$ and thus $s_1 =c^{-1}(s_2) \vee s_1 \prec c^{-1}(s_2)$ by repeated application of \ref{extcflow1}.
 This contradicts the flow-preservation condition of Theorem~\ref{characcausal}.
\end{proof}

\section{Proofs for Section~\ref{s:gflow}}
\label{s:appendix-gflow}

\begin{proof}[Proof of Lemma~\ref{lem:focusing-flow-demand}]
 Lemma~\ref{lem:row-by-col-meaning} implies that for any $u\in\comp{O}$, $(M\v{a})_u = 0$ if and only if
    \begin{align*}
        \triplecase
            {\lambda(u) \in \{ X, XY \}}{u \notin \odd{\someset}}
            {\lambda(u) \in \{ XZ, YZ, Z \}}{u \notin \someset}
            {\lambda(u) = Y}{u \notin \codd{\someset}}
    \end{align*}

 First, assume $\someset$ is focused and consider each focusing condition in turn.
 \begin{itemize}
  \item \ref{F1} states $w \in S\cap\someset \implies \lambda(w) \in \{ XY, X, Y \}$, so if $w\in S$ satisfies $\ld(w)\in\{XZ,YZ,Z\}$, then $w\notin\someset$, which by the above implies $(M\v{a})_w = 0$.
  \item \ref{F2} states $w \in S\cap\odd{\someset} \implies \lambda(w) \in \{ XZ, YZ, Y, Z \}$, so if $w\in S$ satisfies $\ld(w)\in\{X,XY\}$, then $w\notin\odd{\someset}$, which again implies $(M\v{a})_w = 0$.
  \item \ref{F3} states $w\in S$ and $\lambda(w) = Y \Rightarrow w \notin \codd{\someset}$, thus once more $(M\v{a})_w = 0$.
 \end{itemize}
 The vertices in $S$ are non-outputs and -- between the three conditions -- we have considered all the measurement labels, so we may conclude $(M\v{a})_w = 0$ for all $w\in S$ as desired.

 Conversely, assume $(M\v{a})_u = 0$ for all $u\in S$.
 Distinguish cases according to the measurement label of $u$:
 \begin{itemize}
  \item If $\lambda(u) \in \{ X, XY \}$, then $(M\v{a})_u = 0$ implies $u \notin \odd{\someset}$, thus \ref{F2} holds for $u$.
  \item If $\lambda(u) \in \{ XZ, YZ, Z \}$, then $(M\v{a})_u = 0$ implies $u\notin\someset$, thus \ref{F1} holds for $u$.
  \item If $\lambda(u) = Y$, then $(M\v{a})_u = 0$ implies $u \notin \codd{\someset}$, thus \ref{F3} holds for $u$.
 \end{itemize}
 In each case, the other two focusing conditions hold trivially for $u$.
 Hence $\someset$ is focused, as desired.
\end{proof}

Several observations follow straightforwardly from Lemma~\ref{lem:focusing-flow-demand}.

\begin{observation}\label{obs:focusing-properties}
 Let $\Gamma = (G,I,O,\ld)$ be a \LOG. Then:
 \begin{itemize}
  \item If $\someset\sse\comp{I}$ is focused over $\otherset\sse\comp{O}$ and $\otherset'$ is a subset of $\otherset$, then $\someset$ is focused over $\otherset'$.
  \item If $\someset_1,\ldots, \someset_n\sse\comp{I}$ are all focused over $\otherset\sse\comp{O}$, then $\symdi{k=1..n} \someset_k$ is focused over $\otherset$.
   This follows by linearity from noting that the indicator vector of $\symdi{k=1..n} \someset_k$ is the element-wise sum (modulo 2) of the indicator vectors of the individual sets.
  \item If $\someset_1,\someset_2\sse\comp{I}$, then $\someset_1\symd\someset_2$ is focused over the singleton set $\{v\}\sse\comp{O}$ if and only if either both sets are focused over $\{v\}$, or neither is. (This corresponds to \cite[Lemmas~B.2 and~B.3]{simmonsRelatingMeasurementPatterns2021}.)
  \item If a set $\someset\sse\comp{I}$ is focused over sets $\otherset_1\sse\comp{O}$ and $\otherset_2\sse\comp{O}$, then $\someset$ is also focused over $\otherset_1\cup\otherset_2$.
 \end{itemize}
\end{observation}

\begin{proof}[Proof of Theorem~\ref{thm:YZ-into-Pauli}]
 First, suppose $\Gamma$ has a Pauli flow $(c,\prec)$ and there exists a set $K\sse\comp{I}$ which satisfies the three conditions.
 Let $C$ be the correction matrix representing $c$, and let $M, N$ be the flow-demand and order demand-matrices for $\Gamma$.
 Then by Theorem~\ref{thm:algebraic-Pauli} we know that $M C = \I$ and that $N C$ forms a directed acyclic graph.

 The flow-demand and order-demand matrices for $\Gamma'$ can be written in block-matrix form where rows and columns labelled by the original vertices come first and those labelled by $z$ come last.
 Then, with $M$ and $N$ the respective matrices for $\Gamma$, we can write:
 \[
  M_{\Gamma'} = \pmm{M & \v{x} \\ \v{0} & 1}
  \qquad \text{and} \qquad
  N_{\Gamma'} = \pmm{N & \v{z} \\ \v{n} & 0},
 \]
 where
 \begin{itemize}
  \item $\v{x}$ is the indicator vector among the non-outputs for the set of neighbours of $u$ which are measured $X,Y$, or $XY$, \ie $x_w = 1 \Leftrightarrow w\in S\cap\Xlike$.
  \item $\v{z}$ is the indicator vector among the non-outputs for the set of neighbours of $u$ which are measured $YZ$ or $XZ$, \ie $z_w = 1 \Leftrightarrow w\in S\cap\Zlike\cap\planar$.
  \item $\v{n}$ is the indicator vector among the non-inputs for the set of all neighbours of $u$, \ie $n_w = 1 \Leftrightarrow w\in S\setminus I$.
 \end{itemize}
 Let $\v{k}$ be the indicator vector for the set $K$ among the non-inputs and let $C_{\Gamma'}' = \smm{C&\v{k}\\\v{0}&1}$, then
 \[
  \pmm{M & \v{x} \\ \v{0} & 1} \pmm{C & \v{k} \\ \v{0} & 1}
  = \pmm{MC & M\v{k} + \v{x} \\ \v{0} & 1}
  = \pmm{\I & M\v{k} + \v{x} \\ \v{0} & 1}
 \]
 Consider the $v$-th element of $M\v{k}+\v{x}$ for some $v\in\comp{O}$.
 If $v\in\comp{O}\setminus (S\cap\Xlike)$, then $x_v=0$ and assumption~\ref{it:K-focused} together with Lemma~\ref{lem:focusing-flow-demand} implies that $(M\v{k})_v = 0$.
 If $v\in S\cap\Xlike$, then $x_v=1$.
 By assumption~\ref{it:K-focused}, $K$ would not be focused over $(\comp{O}\setminus (S\cap\Xlike))\cup\{v\}$, so by Lemma~\ref{lem:focusing-flow-demand} we must have $(M\v{k})_v = 1$.
 Since $(M\v{k})_v = x_v$ in either case, we find that $M\v{k}+\v{x} = \v{0}$.
 Hence $M_{\Gamma'}C'=\I$, \ie the new correction matrix satisfies the first of the algebraic Pauli flow properties.

 For the product of order-demand and correction matrix, we have:
 \[
  \pmm{N & \v{z} \\ \v{n} & 0} \pmm{C & \v{k} \\ \v{0} & 1}
  = \pmm{NC & N\v{k} + \v{z} \\ \v{n} C & \v{n}\cdot\v{k}}
 \]
 Now, each non-zero term in the sum $\v{n}\cdot\v{k} = \sum_{w\in\comp{I}} n_w k_w$ arises from a vertex $w\in S\cap K$, so $\v{n}\cdot\v{k}\equiv \abs{S\cap K}\bmod 2$.
 Thus, assumption~\ref{it:K-S-even} implies that $\v{n}\cdot\v{k} = 0$.

 Since $M_{\Gamma'}C_{\Gamma'} = \I$, by Lemma~\ref{lem:focusing-flow-demand} applied to each column of $C_{\Gamma'}$ in turn, the matrix $C_{\Gamma'}$ represents a focused correction function on $\Gamma'$.
 Thus, by Observation~\ref{obs:NC-interpretation}, we have
 \begin{itemize}
  \item $(\v{n}C)_w = 1$ for $w\in\comp{O}$ if and only if $z\in\odds{G'}{c(w)}$, which is equivalent to $\abs{c(w)\cap S} \equiv 1 \bmod 2$, and
  \item $(N\v{k} + \v{z})_v = 1$ for $v\in\comp{I}$ if and only if either $\ld(v)=XY$ and $v\in K$, or $\ld(v)\in\{XZ,YZ\}$ and $v\in\odds{G'}{K\cup\{z\}}$.
  As $z\notin K$ and $z\notin\odds{G'}{K}$, we have:
  \[
   \odds{G'}{K\cup\{z\}}
   = \odds{G'}{K\symd\{z\}}
   = \odds{G'}{K}\symd\odds{G'}{\{z\}}
   = \odds{G}{K} \symd S.
  \]
  By focusing, $XY$-measured vertices don't appear in the odd neighbourhood of $K\cup\{z\}$ and other planar-measured vertices do not appear in $K$.
  Thus, equivalently $(N\v{k} + \v{z})_v = 1$ for $v\in\comp{I}$ if and only if $v\in\planar\cap (K\cup (\odds{G}{K} \symd S))$.
 \end{itemize}
 Since we know $NC$ is a DAG and thus extends to a strict partial order $\prec$, Lemma~\ref{lemma_order} together with assumption~\ref{it:order} implies that the product $N_{\Gamma'}C_{\Gamma'}$ also extends to a strict partial order $\prec'$ and hence is a DAG.
 Then, as both $M_{\Gamma'}C_{\Gamma'}$ and $N_{\Gamma'}C_{\Gamma'}$ satisfy the properties of Theorem~\ref{thm:algebraic-Pauli}, the \LOG\ $\Gamma'$ has Pauli flow.

 For the other direction, assume $\Gamma'$ has Pauli flow, \ie there exists a matrix $C_{\Gamma'}'$ such that $M_{\Gamma'}C_{\Gamma'}'=\I$ and $N_{\Gamma'}C_{\Gamma'}'$ is a DAG.
 Express $C_{\Gamma'}'$ in block matrix form (separating out the row and column labelled by $z$, like before) as $\smm{C'&\v{k}'\\\v{h}&d}$, then
 \[
  \I
  = M_{\Gamma'}C_{\Gamma'}'
  = \pmm{M & \v{x} \\ \v{0} & 1} \pmm{C' & \v{k}' \\ \v{h} & d}
  = \pmm{MC' + \v{x}\otimes\v{h} & M\v{k}'+\v{x} \\ \v{h} & d}
 \]
 which implies $M\v{k}'+\v{x}=\v{0}$, $\v{h}=\v{0}$ and $d=1$.
 Then setting $\v{h}=\v{0}$ yields $MC'=\I$ by the properties of block matrices.
 Moreover, $NC'$ is a DAG since a subgraph of a DAG must itself be a DAG, so we immediately have that $\Gamma$ also has Pauli flow.
 It remains to show there exists a set $K$ with the desired properties

 Let $K'\sse\comp{I}$ be the set whose indicator vector is $\v{k}'$.
 Consider first the property $M\v{k}'+\v{x}=\v{0}$.
 Then for all $v\notin S\cap\Xlike$, we have $x_v=0$ and thus $(M\v{k}')_v=0$.
 By Lemma~\ref{lem:focusing-flow-demand}, this means $K'$ is focused over $\comp{O}\setminus (S\cap\Xlike)$, \ie the first part of property~\ref{it:K-focused} holds for $K'$.
 Now, if $v\in S\cap\Xlike$, then $x_v=1$ and thus $(M\v{k}')_v=1$, so by Lemma~\ref{lem:focusing-flow-demand}, $K'$ is not focused over any set that contains $v$.
 This establishes the second part of property~\ref{it:K-focused}.

 Furthermore,
 \[
  N_{\Gamma'}C_{\Gamma'}
  = \pmm{N & \v{z} \\ \v{n} & 0} \pmm{C' & \v{k}' \\ \v{0} & 1}
  = \pmm{NC' & N\v{k}' + \v{z} \\ \v{n} C' & \v{n}\cdot\v{k}'}
 \]
 being a DAG implies that $\v{n}\cdot\v{k}'=0$, which is equivalent to $\abs{S\cap K'}\equiv 0 \bmod 2$: this is property~\ref{it:K-S-even}.
 Finally, if $\abs{c(w)\cap S}\equiv 1\bmod 2$, then $u\in\odds{G'}{c(w)}$ for some $w\in\comp{O}$, so $w\prec z$ by \ref{P2}.
 Similarly, by the definition of $K'$ and the properties of a focused Pauli flow, if $v\in\planar$ satisfies $v\in K' \cup (\odds{G}{K'} \symd S)$, then $v\in K'$ or $v\in\odds{G'}{K'\cup\{z\}}$.
 Then by \ref{P1} and \ref{P2}, $z\prec v$.
 Thus for all such $v,w$ we have $w\prec v$ by transitivity of $\prec$, which implies $\neg(v\prec w \vee v=w)$, \ie property~\ref{it:order}.
 This completes the proof of the reverse direction.
\end{proof}

\begin{corollary}
 Let $\Gamma = (G,I,O,\ld)$ be a \LOG\ and let $S\sse O$, \ie all neighbours of the new vertex are outputs.
 Define $\Gamma'$ to be the \LOG\ that results from inserting a new YZ-measured vertex $z$ with neighbourhood $S$ as in Definition~\ref{def:YZ-insertion-gflow}.
 Then if $\Gamma$ has Pauli flow, $\Gamma'$ also has Pauli flow.
\end{corollary}
\begin{proof}
 Suppose $\Gamma$ has Pauli flow.
 Take $K=\emptyset$, then $K$ is trivially focused over $\comp{O}\setminus S = \comp{O}$ and there is no larger set for which focusing is defined, so condition~\ref{it:K-focused} holds.
 Moreover, condition~\ref{it:K-S-even} is trivially true.
 Finally, condition~\ref{it:order} is also satisfied since the set $\{v\in\planar\mid v\in K\cup\odds{G'}{\{z\}}\}$ is empty.
 Thus by Theorem~\ref{thm:YZ-into-Pauli}, $\Gamma'$ has Pauli flow.
\end{proof}

\begin{corollary}
 Let $\Gamma = (G,I,O,\ld)$ be a \LOG\ and let $S\sse I$, \ie all neighbours of the new vertex are inputs.
 Define $\Gamma'$ to be the \LOG\ that results from inserting a new YZ-measured vertex $z$ with neighbourhood $S$ as in Definition~\ref{def:YZ-insertion-gflow}.
 Then if $\Gamma$ has Pauli flow, $\Gamma'$ also has Pauli flow.
\end{corollary}
\begin{proof}
 Assume $\Gamma$ has focused Pauli flow $(c,\prec)$, let $K = \symdi{s\in S} c(s)$ and recall Observation~\ref{obs:focusing-properties}.
 From the definition of a focused Pauli flow, for each $s\in S$, the correction set $c(s)$ is focused over $\comp{O}\setminus\{s\}$, and thus $c(s)$ is also focused over $\comp{O}\setminus S$.
 Therefore, $K$ is focused over $\comp{O}\setminus S$.
 Moreover, for each $v\in S$ all but one of the components of $K$ are focused over $\{v\}$ and thus $K$ as a whole is not focused over $\{v\}$.
 This establishes condition~\ref{it:K-focused}.
 Condition~\ref{it:K-S-even} is immediate as $K\sse\comp{I}$ by the definition of correction sets.
 By the same reasoning, $\Wpred=\emptyset$ and therefore condition~\ref{it:order} holds.
 Thus by Theorem~\ref{thm:YZ-into-Pauli}, $\Gamma'$ has Pauli flow.
\end{proof}

\begin{proof}[Proof of Corollary~\ref{cor:YZ-insertion-single-neighbour}]
 Given the focused Pauli flow $(c,\prec)$ on $\Gamma$, take $K := c(x)$.
 Firstly, $K$ is focused over $\comp{O}\setminus\{x\}$ since $c$ is a focused correction function, so the first part of condition~\ref{it:K-focused} holds.
 Furthermore, $x\in\odds{G}{c(x)}$, so $K$ is not focused over any set containing $x$, which establishes the second part of condition~\ref{it:K-focused}

 We have $x \notin c(x)$ by \ref{P4}, so $\abs{S\cap K} = \abs{\{x\}\cap c(x)} = 0$ and condition~\ref{it:K-S-even} holds.

 Note that $\Wpred := \{w\in\comp{O} : \abs{c(w)\cap \{x\}}\equiv 1 \bmod 2\} = \{w\in\comp{O}\mid x\in c(w)\}$, so $w\in\Wpred$ implies $w\prec x$ by \ref{P1}.
 Similarly $\Vsucc := \planar\cap (c(x) \cup (\odds{G}{c(x)} \symd \{x\}))$ has the property that $v\in\Vsucc$ implies $x\prec v$ by \ref{P1} and \ref{P2}.
 Hence for all $w\in\Wpred$ and for all $v\in\Vsucc$ we have $w\prec x \prec v$, which implies $\neg(v\prec w \vee v= w)$ since $\prec$ is strict.
 Therefore condition~\ref{it:order} holds.

 Thus, $\Gamma'$ has the given Pauli flow by Theorem~\ref{thm:YZ-into-Pauli}.
\end{proof}

\section{Proofs for Section~\ref{s:other-labels}}
\label{s:appendix-other-labels}

\begin{observation}[Phase shifting]
    If a YZ-measured vertex $z$ is connected to a single XY-measured vertex $x$, then shifting
    part from the phase of $x$ to $z$ is sound.
    \[
    \input{phase-shifting/1.tikz}\hspace*{1em}\eq{}\hspace*{1em}\input{phase-shifting/8.tikz}
    \]
\end{observation}
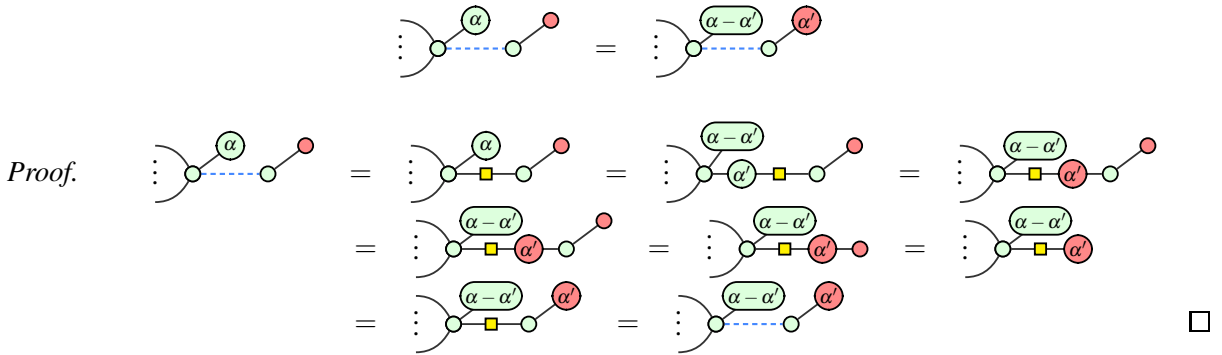
\begin{proof}
    \(
    \hspace*{2em}\input{phase-shifting/1.tikz}\hspace*{1em}\eq{}\hspace*{1em}\input{phase-shifting/2.tikz}
    \hspace*{1em}\eq{}\hspace*{1em}\input{phase-shifting/3.tikz}
    \hspace*{1em}\eq{}\hspace*{1em}\input{phase-shifting/4.tikz}
    \)
    \\
    \(
    \hspace*{12em}\eq{}\hspace*{1em}\input{phase-shifting/4.tikz}
    \hspace*{1em}\eq{}\hspace*{1em}\input{phase-shifting/5.tikz}
    \hspace*{1em}\eq{}\hspace*{1em}\input{phase-shifting/6.tikz}
    \)
    \\
    \(
    \hspace*{12em}\eq{}\hspace*{1em}\input{phase-shifting/7.tikz}
    \hspace*{1em}\eq{}\hspace*{1em}\input{phase-shifting/8.tikz}
    \)
    \end{proof}

\begin{lemma}\label{lem:vertex-splitting-steps}
 The vertex splitting operation of Definition~\ref{def:vertex-splitting} can be effected via the following sequence of steps:
 \begin{enumerate}
  \item Insert a $YZ$-measured vertex $x''$ with neighbourhood $\{x\}$ according to Definition~\ref{def:YZ-insertion-gflow}.
  \item If desired, split measurement angles between $x$ and $x''$.
  \item Insert a $YZ$-measured vertex $x'$ with neighbourhood $W\cup\{x''\}$ according to Definition~\ref{def:YZ-insertion-gflow}.
  \item Pivot on the edge $\{x',x''\}$.
 \end{enumerate}
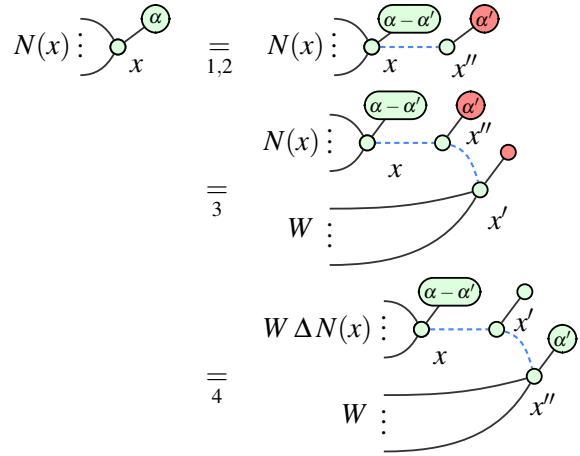
 \begin{align*}
    \input{vertex-splitting/1.tikz}\hspace*{1em}
    &\eq{1,2}\hspace*{1em}\input{vertex-splitting/2.tikz} \\
    &\eq{3}\hspace*{1em}\input{vertex-splitting/3.tikz} \\
    &\eq{4}\hspace*{1em}\input{vertex-splitting/4.tikz}
\end{align*}
 The sequence can also be reversed.
\end{lemma}

\begin{proof}
 Let $\Gamma_0:=\Gamma$ be the initial \LOG, with $G=(V,E)$, and denote by $\Gamma_k = (G_k, I, O, \ld_k)$ the \LOG\ after step $k$ in the above list.
 Then the \LOG{}s after each step have the following graphs and measurement labellings:
 \begin{enumerate}
  \item $G_1 = (V\cup\{x''\}, E\cup\{\{x,x''\}\})$ and $\ld_1$ is the extension of $\ld$ that satisfies $\ld_1(x'')=YZ$.

  \item $\Gamma_2=\Gamma_1$ as changing the measurement angles of planar-measured vertices does not affect the underlying \LOG.

  \item $G_3 = (V_3,E_3)$ where $V_3 = V\cup\{x',x''\}$ and $E_3 = E\cup\{\{x,x''\}, \{x',x''\}\}\cup (\{x'\}\times W)$, with $\ld_3$ being the extension of $\ld_1$ that satisfies $\ld_3(x')=YZ$.
  Note that all the new edges involve at least one vertex that did not appear in the original graph and hence are definitely not contained in $E$.
  Thus we may equivalently write $E_3 = E\symd\{\{x,x''\}, \{x',x''\}\}\symd (\{x'\}\times W)$.

  \item By Lemma~\ref{lem:pivot-odd-neighbourhoods} applied to singleton sets (so that the odd neighbourhood reduces to the usual neighbourhood), we have
  \[
   N_{G_3\wedge x'x''}(v) =
   \begin{cases}
    N_{G_3}(v) &\text{if } x',x''\notin N_{G_3}[v] \\
    N_{G_3}(v) \symd N_{G_3}[x''] &\text{if } x'\in N_{G_3}[v] \wedge x''\notin N_{G_3}[v] \\
    N_{G_3}(v) \symd N_{G_3}[x'] &\text{if } x'\notin N_{G_3}[v] \wedge x''\in N_{G_3}[v] \\
    N_{G_3}(v) \symd N_{G_3}[x'] \symd N_{G_3}[x''] &\text{if } x',x''\in N_{G_3}[v]
   \end{cases}
  \]
  or, by noting that $u\in N_{G_3}[v] \Leftrightarrow v\in N_{G_3}[u]$ and $N_{G_3}[x'] = \{x', x''\}\cup W$ and $N_{G_3}[x''] = \{x,x',x''\}$:
  \begin{equation}\label{eq:neighbourhoods-pivoting}
   N_{G_4}(v) =
   N_{G_3\wedge x'x''}(v) =
   \begin{cases}
    N_{G_3}(v) &\text{if } v\notin\{x,x',x''\}\cup W \\
    N_{G_3}(v) \symd \{x, x', x''\} &\text{if } v\in W \\
    N_{G_3}(v) \symd \{ x', x''\} \symd W &\text{if } v = x \\
    N_{G_3}(v) \symd \{x\} \symd W &\text{if } v\in\{x',x''\}
   \end{cases}
  \end{equation}
  This implies we toggle edges between $W$ and $\{x,x',x''\}$ as well as edges between $x$ and $\{x',x''\}$.
  (Note that each toggled edge appears twice in \eqref{eq:neighbourhoods-pivoting}, once for each endpoint.)
  Thus, $G_4 = (V_4,E_4)$ where $V_4 = V_3 = V\cup\{x',x''\}$ and:
  \begin{align*}
   E_4
   &= E_3 \symd (\{x,x',x''\}\times W) \symd \{\{x,x'\},\{x,x''\}\} \\
   &= E\symd\{\{x,x''\}, \{x',x''\}\} \symd (\{x'\}\times W) \symd (\{x,x',x''\}\times W) \symd \{\{x,x'\},\{x,x''\}\} \\
   &= E\symd\{\{x,x'\}, \{x',x''\}\} \symd (\{x,x''\}\times W)
  \end{align*}
  Furthermore, $\ld_4$ satisfies $\ld_4(x')=\ld_4(x'')=XY$ and $\ld_4(v)=\ld(v)$ for all other non-outputs.
 \end{enumerate}
 It is straightforward to check that $\Gamma_4$ is the same as the \LOG\ $\Gamma'$ of Definition~\ref{def:vertex-splitting}.

 For the reverse sequence, note that pivoting is involutive and that $YZ$-measured vertices can be deleted without affecting the interpretation of the underlying diagram if their measurement angle is~0.
\end{proof}

\begin{remark}\label{rem:pivoting}
    It is well-known {\cite{duncanGraphtheoreticSimplificationQuantum2020}} that in ZX-calculus, pivoting on edge $\{u,v\}$ reads:
    \[
    \input{pivoting/1.tikz}\hspace*{1em}\eq{\wedge uv}\hspace*{1em}\input{pivoting/2.tikz}
    \]
    What happens to the measurement effects when pivoting on an edge $\{u,v\}$ can be derived from this
    and rewrite rules. For instance, as $\input{pivoting/5.tikz}\hspace*{0.5em}\eq{}\hspace*{0.5em}\input{pivoting/6.tikz}$, if $u$ and $v$ were YZ-measured
    before the pivot, they will be XY-measured afterwards.
\end{remark}

\begin{lemma}[{\cite[Lemma~B.9]{duncanGraphtheoreticSimplificationQuantum2020}}]\label{lem:pivot-odd-neighbourhoods}
 Given a graph $G=(V,E)$, $A\sse V$, $u\in V$, and $v\in N_G(u)$,
 \[
  \odds{G\wedge uv}{A} = \begin{cases}
                         \odds{G}{A} &\text{if } u,v\notin\codds{G}{A} \\
                         \odds{G}{A}\symd N_G[v] &\text{if } u\in\codds{G}{A}, v\notin\codds{G}{A} \\
                         \odds{G}{A}\symd N_G[u] &\text{if } u\notin\codds{G}{A}, v\in\codds{G}{A} \\
                         \odds{G}{A}\symd N_G[u]\symd N_G[v] &\text{if } u,v\in\codds{G}{A}.
                        \end{cases}
 \]
\end{lemma}

\begin{proof}[Proof of Lemma~\ref{lem:pivot-Pauli-flow}]
 First, note (e.g.\ via Remark~\ref{rem:pivoting}) that $u,v\in N_G[u]$, that symmetrically $u,v\in N_G[v]$, and that
 \[
  \codds{G}{A\symd\{u\}} = \codds{G}{A}\symd N_G[u].
 \]
 Thus, for any set $A\sse V$, we have $u\in\codds{G}{A\symd\{u\}} \Leftrightarrow u\notin\codds{G}{A}$ and similarly $v\in\codds{G}{A\symd\{u\}} \Leftrightarrow v\notin\codds{G}{A}$.

 By Lemma~\ref{lem:pivot-odd-neighbourhoods}, in the four cases (depending on which of $u,v$ are in $\codds{G}{c(w)}$), we have:
 \begin{itemize}
  \item If $u,v\notin\codds{G}{c(w)}$, then
   \[
    \odds{G\wedge uv}{c'(w)} = \odds{G\wedge uv}{c(w)} = \odds{G}{c(w)}
   \]
  \item If $u\in\codds{G}{c(w)}$ and $v\notin\codds{G}{c(w)}$, then $u\notin\codds{G}{c(w)\symd\{u\}}$ and $v\in\codds{G}{c(w)\symd\{u\}}$; thus
   \begin{align*}
    \odds{G\wedge uv}{c'(w)}
    &= \odds{G\wedge uv}{c(w)\symd\{u\}} \\
    &= \odds{G}{c(w)\symd\{u\}}\symd N_G[u] \\
    &= \odds{G}{c(w)}\symd N_G(u)\symd N_G[u] \\
    &= \odds{G}{c(w)}\symd\{u\}
   \end{align*}
  \item If $u\notin\codds{G}{c(w)}, v\in\codds{G}{c(w)}$, then symmetrically to the previous case,
   \[
    \odds{G\wedge uv}{c'(w)}
    = \odds{G\wedge uv}{c(w)\symd\{v\}}
    = \odds{G}{c(w)}\symd\{v\}
   \]
  \item If $u,v\in\codds{G}{c(w)}$, then $u,v\in\codds{G}{c(w)\symd\{u,v\}}$ and thus
   \begin{align*}
    \odds{G\wedge uv}{c'(w)}
    &= \odds{G\wedge uv}{c(w)\symd\{u,v\}} \\
    &= \odds{G}{c(w)\symd\{u,v\}}\symd N_G[u]\symd N_G[v] \\
    &= \odds{G}{c(w)}\symd N_G(u)\symd N_G(v)\symd N_G[u]\symd N_G[v] \\
    &= \odds{G}{c(w)}\symd\{u,v\}
   \end{align*}
 \end{itemize}
 These results are summarised in Table~\ref{tab:pivoting-codd}.
 They straightforwardly imply $\codds{G\wedge uv}{c'(w)} = \codds{G}{c(w)}$ for all $w\in\comp{O}$: in each row, the changes to correction set and odd neighbourhood are the same, so they cancel out when computing the closed odd neighbourhood.

 Note that the only changes to correction sets and odd neighbourhoods are symmetric differences with some subset of $\{u,v\}$, so Pauli flow and (where applicable) focusing conditions involving only vertices other than $u,v$ remain satisfied.

 \begin{table}
  \centering
  \begin{tabular}{l||l|l}
   Case & $c'(w)$ & $\odds{G\wedge uv}{c'(w)}$ \\ \hline
   $u,v\notin\codds{G}{c(w)}$ & $c(w)$ & $\odds{G}{c(w)}$ \\
   $u\in\codds{G}{c(w)}, v\notin\codds{G}{c(w)}$ & $c(w)\symd\{u\}$ & $\odds{G}{c(w)}\symd\{u\}$ \\
   $u\notin\codds{G}{c(w)}, v\in\codds{G}{c(w)}$ & $c(w)\symd\{v\}$ & $\odds{G}{c(w)}\symd\{v\}$ \\
   $u,v\in\codds{G}{c(w)}$ & $c(w)\symd\{u,v\}$ & $\odds{G}{c(w)}\symd\{u,v\}$
  \end{tabular}
  \caption{Correction sets and odd neighbourhoods after pivoting in Lemma~\ref{lem:pivot-Pauli-flow}, depending on the relationship of $u$ and $v$ to the original closed odd neighbourhood.}\label{tab:pivoting-codd}
 \end{table}

 Additionally, whether a symmetric difference is taken with $u$ (and what effect this may have) is independent of what is done regarding $v$.
 By symmetry, it thus suffices to consider only $u$.
 First, consider the correction set for $u$ itself.
 \begin{itemize}
  \item If $\ld(u)=XY$, then $u\notin c(u) \wedge u\in\odds{G}{c(u)}$ by \ref{P4}, so $u\in\codds{G}{c(u)}$.
  Thus we take symmetric differences with $u$, which means $u\in c'(u)\wedge u\notin\odds{G\wedge uv}{c'(u)}$.
  Yet note that $\ld'(u)=YZ$ and the above is exactly what is needed to satisfy \ref{P6}.
  \item The case $\ld(u)=YZ$ is symmetric to $\ld(u)=XY$.
  \item If $\ld(u)=XZ$, then $u\in c(u) \wedge u\in\odds{G}{c(u)}$, so $u\notin\codds{G}{c(u)}$.
  Thus the membership of $u$ in those sets does not change and \ref{P5} remains satisfied.
  \item If $\ld(u)=X$, then $u\in\odds{G}{c(u)}$ by \ref{P7}, and $\ld'(u)=Z$.
  There are two subcases:
   \begin{itemize}
    \item If furthermore $u\in c(u)$, then $u\notin\codds{G}{c(u)}$ and nothing changes, meaning $u\in c'(u)$ and $u\in\odds{G\wedge uv}{c'(u)}$.
    Hence \ref{P9} is satisfied.
    \item If instead $u\notin c(u)$, then $u\in\codds{G}{c(u)}$.
    Thus symmetric differences with $u$ are taken and we have $u\in c'(u)$ and $u\notin\odds{G\wedge uv}{c(u')}$.
    Again, \ref{P9} is satisfied.
   \end{itemize}
  \item The case $\ld(u)=Z$ is symmetric to $\ld(u)=X$.
  \item If $\ld(u)=Y$, then there are two possibilities but in either case $u\in\codds{G}{c(u)}$ by \ref{P8}.
  Thus symmetric differences are taken, and afterwards $u\in\codds{G\wedge uv}{c'(u)}$ again.
  The measurement label remains $\ld'(u)=Y$ and \ref{P8} remains satisfied.
 \end{itemize}
Now suppose $w\in\comp{O}\setminus\{u\}$.
Note that by the first part of the proof,
\begin{equation}\label{eq:pivot-changes}
 u\in c'(w) \Leftrightarrow u\in\odds{G}{c(w)} \quad\text{and}\quad
 u\in\odds{G\wedge uv}{c'(w)} \Leftrightarrow u\in c(w).
\end{equation}
If $u\notin c(w)$ and $u\notin\odds{G}{c(w)}$, then analogous properties will hold after pivoting, so there is nothing to prove.
We will therefore only consider cases where $u$ is in the correction set, the odd neighbourhood, or both.
\begin{itemize}
 \item If $\ld(u)=XY$, and $u$ appears in $c(w)$ and/or $\odds{G}{c(w)}$, then $w\prec u$ by \ref{P1} and \ref{P2}.
 So, whenever $u$ appears in $c'(w)$ and/or $\odds{G\wedge uv}{c'(w)}$, we have $w\prec u$ and thus \ref{P1} and \ref{P2} remain satisfied for $c'(w)$.
 \item The case $\ld(u)=YZ$ is analogous to $\ld(u)=XY$.
 \item The case $\ld(u)=XZ$ is also analogous to $\ld(u)=XY$.
 \item If $\ld(u)=X$, there are two subcases.
  \begin{itemize}
   \item Suppose $u\in\odds{G}{c(w)}$, then $w \prec u$.
    Then $u\in c'(w)$ but \ref{P1} remains satisfied after the change of label to $\ld'(u)=Z$ since we still have $w\prec u$.
   \item Suppose $u\in c(w)$ and $u\notin\odds{G}{c(w)}$.
    In this case, nothing is implied about the order of $u$ and $w$.
    After the pivot, we have $\ld'(u)=Z$, $u\notin c'(w)$ and $u\in\odds{G\wedge uv}{c'(w)}$.
    Thus, \ref{P2} is satisfied.
  \end{itemize}
 \item The case $\ld(u)=Z$ is symmetric to $\ld(u)=X$.
 \item If $\ld(u)=Y$, there are again two cases.
  \begin{itemize}
   \item If $u\notin\codds{G}{c(w)}$, then we know nothing about the order of $u$ and $w$ because of \ref{P3}, but nothing changes regarding the correction set or the label of $u$, so \ref{P3} remains satisfied.
   \item If $u\in\codds{G}{c(w)}$, then by \ref{P3}, we have $w\prec u$.
    In this case $u$ moves from the correction set to the odd neighbourhood, or conversely, but \ref{P3} remains satisfied.
  \end{itemize}
\end{itemize}
The arguments regarding $v$ are independent of and analogous to the ones for $u$, so this completes the proof that $(p',\prec)$ is a Pauli flow.

For focusing, it suffices to consider the correction sets of vertices other than $u$ (as focusing does not talk about the relationship between a vertex and its own correction set or the latter's odd neighbourhood).
So suppose $w\in\comp{O}\setminus\{u\}$.
\begin{itemize}
 \item If $\ld(u)=Y$ then $(p,\prec)$ being focused implies $u\notin\codds{G}{c(w)}$ via \ref{F3}.
  But then nothing changes during the pivot, not even the measurement label of $u$, so \ref{F3} remains satisfied.
 \item If $\ld(u)=XY$ or $\ld(u)=X$, then $(p,\prec)$ being focused implies $u\notin\odds{G}{c(w)}$ via \ref{F2}.
  After the pivot, $\ld'(u)\in\{YZ, Z\}$ and $u\notin {c'(w)}$ by \eqref{eq:pivot-changes}, so \ref{F1} is satisfied.
 \item If $\ld(u)=YZ$ or $\ld(u)=Z$, the argument is symmetric to the previous case.
 \item The case $\ld(u)=XZ$ is excluded.
\end{itemize}
Thus $(p,\prec)$ being focused implies that $(p',\prec)$ is focused, assuming the vertices incident on the pivot edge are not $XZ$-measured.

Finally, since the gflow conditions are a subset of the Pauli flow conditions and pivoting maps planar measurements to planar measurements, if $(p,\prec)$ is a gflow, then $(p',\prec')$ is also a gflow.
\end{proof}

\begin{remark}
 In Lemma~\ref{lem:pivot-Pauli-flow}, focusing may break if one of the vertices incident on the pivot edge is $XZ$-measured, say $\ld(u)=XZ$: assume $(p,\prec)$ is focused, then we know $u\notin c(w)$ for any $w\neq u$ but we may have $u\in\odds{G}{c(w)}$.
 In the latter case, we would have $u\in c'(w)$ after the pivot but $\ld(u)=XZ$ still, so the flow is no longer focused.
 Of course it can be re-focused by replacing $c'(w)$ with $c'(w)\symd c'(u)$ for all $w$ such that $u\in\odds{G}{c(w)}$, but the straightforward modification of Lemma~\ref{lem:pivot-Pauli-flow} is no longer sufficient.
\end{remark}

\begin{proof}[Proof of Theorem~\ref{thm:vertex-splitting}]
 By Lemma~\ref{lem:vertex-splitting-steps}, the vertex splitting operation of Definition~\ref{def:vertex-splitting} can be expressed as a sequence of two $YZ$-insertions and a pivot, as well as an operation on measurement angles that does not affect the \LOG{}s.
 As $\abs{I}=\abs{O}$, the focused Pauli flow on each \LOG\ is unique.
 Using the same numbering as in the above lemma, $\Gamma_1$ has Pauli flow if and only if $\Gamma$ has Pauli flow by Corollary~\ref{cor:YZ-insertion-single-neighbour}.
 Similarly, by Lemma~\ref{lem:pivot-Pauli-flow} and the property that pivoting is involutive, $\Gamma_4$ has Pauli flow if and only if $\Gamma_3$ has Pauli flow.
 In each case, there are no conditions on the flows.
 Thus it remains only to show that step~3 -- the second $YZ$-insertion -- preserves the existence of Pauli flow under the given conditions.

 First, suppose $\Gamma$ has focused Pauli flow $(c,\prec)$ where without loss of generality $\prec$ is the induced partial order.
 Then, by Corollary~\ref{cor:YZ-insertion-single-neighbour}, the unique focused Pauli flow on $\Gamma_2 = \Gamma_1$ satisfies
 \[
  {\prec_2}
  = {\prec} \cup \{(w,x'')\mid w\prec x\} \cup \{(x'',v)\mid x\prec v\}
  \quad\text{and}\quad
  c_2(v)
  =
  \begin{cases}
   c(x)\cup\{x''\} &\text{if } v = x'' \\
   c(v) &\text{otherwise.}
  \end{cases}
 \]
 Furthermore, by Theorem~\ref{thm:YZ-into-Pauli}, $\Gamma_3$ has focused Pauli flow if and only if there exists $K\sse V_2 = V\cup\{x''\}$ such that three conditions hold with respect to the neighbourhood $S=W\cup\{x''\}$ of the new vertex.
 We now show those conditions follow from \ref{VS1}--\ref{VS4}.
 \begin{itemize}
  \item The only difference between $\Gamma_2$ and $\Gamma$ is the insertion of the $YZ$-measured vertex $x''$ with a single edge to $x$.
   Thus $K$ being focused over $\comp{O}\setminus(W\cap\Xlike)$ in $\Gamma$ by \ref{VS1} implies $K$ being focused over $(V\setminus O)\setminus(W\cap\Xlike_2)$ in $\Gamma_2$.

   Note that $x''\notin K$ by assumption, so $K$ is focused over $\{x''\}$.
   Therefore by Observation~\ref{obs:focusing-properties}, $K$ is focused over $(V\cup\{x''\}\setminus O)\setminus(W\cap\Xlike_2)$ in $\Gamma_2$.
   We can also add $x''$ to $W$ since $x''\notin\Xlike_2$.
   Hence $K$ is focused over $\comp{O}_2\setminus((W\cup\{x''\})\cap\Xlike_2)$, which is the first part of condition~\ref{it:K-focused}.
   The second part follows similarly from \ref{VS1}.

  \item We have $\abs{W\cap K}\equiv 0\bmod 2$ by \ref{VS1} and $x''\notin K$.
  Thus $\abs{(W\cup\{x''\})\cap K}\equiv 0\bmod 2$, which is condition~\ref{it:K-S-even}.

  \item Denote by $\Wpred^{(2)}$ and $\Vsucc^{(2)}$ the sets from condition~\ref{it:order}, \ie
  \begin{align}
   \Wpred^{(2)}
   &= \left\{w\in\comp{O}_2: \abs{c(w)\cap (W\cup\{x''\})}\equiv 1 \bmod 2\right\} \label{eq:Wpred2-def}\\
   \Vsucc^{(2)}
   &= \planar_2 \cap (K \cup (\odds{G_2}{K} \symd (W\cup\{x''\}))) \label{eq:Vsucc2-def}
  \end{align}
   Note that $x''$ does not appear in any correction set other than its own because of focusing, so the first definition reduces to
   \begin{equation}\label{eq:Wpred}
    \Wpred^{(2)} =
    \begin{cases}
     \Wpred &\text{if } \abs{c(x)\cap W} \equiv 1\bmod 2 \\
     \Wpred \cup\{x''\} &\text{if } \abs{c(x)\cap W} \equiv 0\bmod 2
    \end{cases}
   \end{equation}
   Similarly, since $x''\notin K$, the set $\odds{G_2}{K}$ is equal to $\odds{G}{K}$ if $x\notin K$ and to $\odds{G}{K}\cup\{x''\}$ if $x\in K$.
   Thus:
   \begin{equation}\label{eq:Vsucc}
    \Vsucc^{(2)} =
    \begin{cases}
     \Vsucc &\text{if } x\in K \\
     \Vsucc \cup \{x''\} &\text{if } x\notin K
    \end{cases}
   \end{equation}
   We distinguish cases depending on the parity of $\abs{c(x)\cap W}$ and on whether $x$ is in $K$.
   \begin{itemize}
    \item Suppose $\abs{c(x)\cap W} \equiv 1\bmod 2$ and $x\in K$.
    This is forbidden by \ref{VS4} so can be ignored.
    \item Suppose $\abs{c(x)\cap W} \equiv 0\bmod 2$ and $x\in K$.
     Take $w\in\Wpred^{(2)}$.

     If $w=x''$, then for all $u\in V$ such that $u\prec_2 x''$ we also have $u\prec x$ by the definition of $\prec_2$, and thus by \ref{VS3} $u\notin\Vsucc^{(2)}$.
     Moreover, $x''\notin\Vsucc^{(2)}$ since $x\in K$.
     Thus $\neg(v\prec_2 w \vee v = w)$ for all $v\in\Vsucc^{(2)}$.

     If $w\neq x''$, the result follows from \ref{VS2}, noting that for all $u,v\in V$ we have $u\prec_2 v$ if and only if $u\prec v$.

    \item Suppose $\abs{c(x)\cap W} \equiv 1\bmod 2$ and $x\notin K$.
    This is again forbidden by \ref{VS4}.

    \item Suppose $\abs{c(x)\cap W} \equiv 0\bmod 2$ and $x\notin K$.
     Take $v\in\Vsucc^{(2)}$.

     If $v=x''$, then for all $u\in V$ such that $x''\prec_2 u$ we also have $x\prec u$ by the definition of $\prec_2$, and thus by \ref{VS3} $u\notin\Wpred^{(2)}$.
     Moreover, $x''\notin\Wpred^{(2)}$ since $\abs{c(x)\cap W} \equiv 0\bmod 2$.
     Thus $\neg(v\prec_2 w \vee v = w)$ for all $w\in\Wpred^{(2)}$.

     If $v\neq x''$, the result again follows from \ref{VS2} and noting that for all $u,w\in V$ we have $u\prec_2 w$ if and only if $u\prec w$.
   \end{itemize}
   Thus condition~\ref{it:order} follows from \ref{VS2}--\ref{VS4}.
 \end{itemize}
 We have shown that $\Gamma'$ has Pauli flow if $\Gamma$ has a focused Pauli flow and we can find a set $K$ such that \ref{VS1}--\ref{VS4} hold.

 For the reverse direction, assume $\Gamma'$ has focused Pauli flow $(c_4,\prec_4)$, then by Lemma~\ref{lem:pivot-Pauli-flow}, $\Gamma_3$ has focused Pauli flow $(c_3,\prec_3)$, where ${\prec_3}={\prec_4}$ and for all $w\in V_4 = V\cup\{x',x''\}$:
 \[
  c_3(w) := \begin{cases}
            c_4(w) &\text{if } x',x''\notin\codds{G_4}{c_4(w)} \\
            c_4(w)\symd\{x'\} &\text{if } x'\in\codds{G_4}{c_4(w)}, x''\notin\codds{G_4}{c_4(w)} \\
            c_4(w)\symd\{x''\} &\text{if } x'\notin\codds{G_4}{c_4(w)}, x''\in\codds{G_4}{c_4(w)} \\
            c_4(w)\symd\{x',x''\} &\text{if } x',x''\in\codds{G_4}{c_4(w)}
           \end{cases}
 \]
 Take $K' = c_3(x')\setminus\{x'\}$, \ie $K' = c_4(x')\setminus\{x''\}$ if $x''\in c_4(x')$ and $K' = c_4(x')$ otherwise.
 By Theorem~\ref{thm:YZ-into-Pauli}, we know that $\Gamma_2$ (and thus iteratively $\Gamma$) has Pauli flow, and conditions~\ref{it:K-focused}--\ref{it:order} hold for $\Gamma_2$.
 \begin{enumerate}
  \item $K'$ being focused over $\comp{O}_2\setminus((W\cup\{x''\})\cap\Xlike_2)$ and no larger set implies $K'$ is focused over $(\comp{O}\cup\{x''\})\setminus(W\cap\Xlike)$ since $\ld_2(x'') = \YZ$.
  Now $x''\notin K'$ by definition, so \ref{F1}--\ref{F3} hold with respect to $x''$ and we may limit attention to the remaining vertices.
  This yields the condition of $K'$ being focused over $\comp{O}\setminus(W\cap\Xlike)$ and no larger set, which is exactly the definition of $K$.
  Moreover, $\abs{(W\cup\{x''\})\cap K'} \equiv 0 \bmod 2$ implies $\abs{W\cap K'} \equiv 0 \bmod 2$ since $x''\notin K'$ by the definition above.
  Thus \ref{VS1} holds.

  \item Let $\Wpred^{(2)}$ and $\Vsucc^{(2)}$ as defined in \eqref{eq:Wpred2-def} and\eqref{eq:Vsucc2-def}.
  Then all $w\in\Wpred^{(2)}$ and all $v\in\Vsucc^{(2)}$ satisfying $\neg(v\prec_2 w \vee v = w)$ immediately implies \ref{VS2} via \eqref{eq:Wpred} and \eqref{eq:Vsucc}, as well as noting that for vertices $u,u'\in V$, we have $u\prec_2 u' \Leftrightarrow u\prec u'$.

  Then from \eqref{eq:Wpred}, we see that $x''\in\Wpred^{(2)}$ if and only if $\abs{c(x)\cap W}\equiv 0 \bmod 2$ and $x\in\Wpred^{(2)}$ if and only if $\abs{c(x)\cap W}\equiv 1 \bmod 2$.
  Similarly, from \eqref{eq:Vsucc}, we deduce that $x''\in\Vsucc^{(2)}$ if and only if $x\notin K$ and $x\in\Vsucc^{(2)}$ if and only if $x\in K$.
  Thus the requirement that $\Wpred^{(2)}$ and $\Vsucc^{(2)}$ do not share elements implies either $\abs{c(x)\cap W}\equiv 0 \bmod 2$ and $x\in K$, or $\abs{c(x)\cap W}\equiv 1 \bmod 2$ and $x\notin K$: \ie \ref{VS4} holds.
  It also means that exactly one of $x,x''$ is in $\Wpred^{(2)}$ and the other is in $\Vsucc^{(2)}$.

  Finally, by the uniqueness of focused Pauli flow on \LOG{}s with $\abs{I}=\abs{O}$ and by Corollary~\ref{cor:YZ-insertion-single-neighbour}, we know that $x$ and $x''$ must have the same partial order relationships in $\prec_2$.
  Suppose $x\in\Wpred^{(2)}$ and $x''\in\Vsucc^{(2)}$ (the other case is analogous).
  Then for all $v\in\Vsucc^{(2)}$, condition~\ref{it:order} implies $\neg(v\prec_2 x)$.
  So if $u\in V$ satisfies $u\prec x$ and thus $u\prec_2 x$, then $u\notin\Vsucc^{(2)}$.
  Similarly, for all $w\in\Wpred^{(2)}$, condition~\ref{it:order} implies $\neg(v\prec_2 x'')$.
  So if $u\in V$ satisfies $x\prec u$ and thus $x''\prec_2 u$, then $u\notin\Wpred^{(2)}$.
  Thus \ref{VS3} holds.
 \end{enumerate}
 This establishes all six conditions hold if $\Gamma'$ has Pauli flow.
\end{proof}

\begin{proof}[Proof of Corollary~\ref{cor:neighbour-unfusion}]
 Suppose $\Gamma$ has a focused Pauli flow $(c,\prec)$ which satisfies the two conditions and where without loss of generality $\prec$ is the induced partial order.
 Take $K := c(b)$ in Theorem~\ref{thm:vertex-splitting}.
 Then focusing of $(c,\prec)$ implies $K$ is focused over $\comp{O}\setminus\{b\}$ and \ref{P4} implies $K$ is not focused over $\{b\}$ as well as $\abs{\{b\}\cap c(b)}=0\bmod 2 \Leftrightarrow b\notin c(b)$.
 Together these imply \ref{VS1}.
 We also have
 \begin{align*}
  \Wpred
  := \{w\in\comp{O} : \abs{c(w)\cap\{b\}}\equiv 1 \bmod 2\}
  = \{w\in\comp{O}\mid b\in c(w)\}
  = \{w\in\comp{O}\mid w\trl_c b\}
 \end{align*}
 and
 \begin{align*}
  \Vsucc
  := \mathcal{L}\cap (c(b)\cup(\odds{G}{c(b)}\symd \{b\}))
  = \{v\in\comp{O}\mid b\trl_c v\}.
 \end{align*}
 Thus \ref{VS2} follows directly from $\prec$ being the induced partial order (cf.\ Definition~\ref{def:trl_c}).
 Next, \ref{VS3} reduces to $a\prec u\implies \neg(u\trl_c b)$ and $u\prec a\implies \neg(b\trl_c u)$ for all $u\in V$, which follows from the second bullet point.
 Finally, \ref{VS4} becomes $a\in c(b) \Leftrightarrow b\notin c(a)$, which is equivalent to the first bullet point.
 Thus all the conditions of Theorem~\ref{thm:vertex-splitting} are satisfied and the operation preserves the existence of Pauli flow.

 For the other direction, suppose $\Gamma'$ has Pauli flow.
 Then by Theorem~\ref{thm:vertex-splitting}, $\Gamma$ has a focused Pauli flow $(c,\prec)$ such that conditions~\ref{VS1}--\ref{VS4} hold.
 We may assume without loss of generality that $\prec$ is the induced partial order (cf.\ Section~\ref{s:further-properties}), since removing relationships from the partial order will not break any of the four conditions.
 Then, from the above, \ref{VS4} implies the first bullet point.
 Now suppose for a contradiction $u\in V$ satisfies $b\prec u\prec a$.
 Then there exists a sequence $u_1,\ldots,u_n$ such that $b\trl_c u_1 \trl_c \ldots \trl_c u_n \trl_c u$ since $\prec$ is the induced partial order.
 Thus there exists $u_1\in\Vsucc$ such that $u_1\prec a$, which contradicts condition~\ref{VS3} so this cannot happen and we have $u\prec a\implies \neg(b\prec u)$.
 A similar argument works for the case $a\prec u\prec b$.
 Thus the second bullet point holds.
\end{proof}

\end{document}

%% file: zx.tikzstyles
\tikzstyle{box}=[shape=rectangle, text height=1.5ex, text depth=0.25ex, yshift=0.5mm, fill=white, draw=black, minimum height=5mm, yshift=-0.5mm, minimum width=5mm, font={\small}]
\tikzstyle{Z dot}=[inner sep=0mm, minimum size=2mm, shape=circle, draw=black, fill={rgb,255: red,221; green,255; blue,221}]
\tikzstyle{Z phase dot}=[minimum size=1.2em, font={\footnotesize\boldmath}, shape=rectangle, rounded corners=0.5em, inner sep=0.2em, outer sep=-0.2em, scale=0.8, tikzit shape=circle, draw=black, fill={rgb,255: red,221; green,255; blue,221}, tikzit draw=blue]
\tikzstyle{X dot}=[Z dot, shape=circle, draw=black, fill={rgb,255: red,255; green,136; blue,136}]
\tikzstyle{X phase dot}=[Z phase dot, tikzit shape=circle, tikzit draw=blue, fill={rgb,255: red,255; green,136; blue,136}, font={\footnotesize\boldmath}]
\tikzstyle{hadamard}=[fill=yellow, draw=black, shape=rectangle, inner sep=0.6mm, minimum height=1.5mm, minimum width=1.5mm]
\tikzstyle{vertex}=[inner sep=0mm, minimum size=1mm, shape=circle, draw=black, fill=black]
\tikzstyle{vertex set}=[inner sep=0mm, minimum size=1mm, shape=circle, draw=black, fill=white, font={\footnotesize\boldmath}]
\tikzstyle{target}=[inner sep=0mm, minimum size=3mm, shape=circle, draw=black]

\tikzstyle{hadamard edge}=[-, dashed, dash pattern=on 2pt off 1.5pt, thick, draw={rgb,255: red,68; green,136; blue,255}]
\tikzstyle{brace edge}=[-, tikzit draw=blue, decorate, decoration={brace,amplitude=1mm,raise=-1mm}]
\tikzstyle{diredge}=[->]
\tikzstyle{dashed edge}=[-, dashed, dash pattern=on 2pt off 0.5pt, draw=black]

%% file: figures/def/def_zspider.tikz
\begin{tikzpicture}
	\begin{pgfonlayer}{nodelayer}
		\node [style=Z phase dot] (0) at (0, 0) {$\alpha$};
		\node [style=none] (1) at (1.25, 0.5) {};
		\node [style=none] (2) at (1.25, -0.75) {};
		\node [style=none] (3) at (1.25, 0.125) {$\vdots$};
		\node [style=none] (4) at (1.25, 0.75) {};
		\node [style=none] (5) at (-1.25, 0.5) {};
		\node [style=none] (6) at (-1.25, -0.75) {};
		\node [style=none] (7) at (-1.25, 0.125) {$\vdots$};
		\node [style=none] (8) at (-1.25, 0.75) {};
	\end{pgfonlayer}
	\begin{pgfonlayer}{edgelayer}
		\draw [in=-180, out=-45] (0) to (2.center);
		\draw [in=180, out=45] (0) to (4.center);
		\draw [in=180, out=30, looseness=1.25] (0) to (1.center);
		\draw [in=135, out=0] (8.center) to (0);
		\draw [in=150, out=0] (5.center) to (0);
		\draw [in=-135, out=0] (6.center) to (0);
	\end{pgfonlayer}
\end{tikzpicture}

%% file: figures/def/def_xspider.tikz
\begin{tikzpicture}
	\begin{pgfonlayer}{nodelayer}
		
		\node [style=none] (1) at (1.25, 0.5) {};
		\node [style=none] (2) at (1.25, -0.75) {};
		\node [style=none] (3) at (1.25, 0.125) {$\vdots$};
		\node [style=none] (4) at (1.25, 0.75) {};
		\node [style=none] (5) at (-1.25, 0.5) {};
		\node [style=none] (6) at (-1.25, -0.75) {};
		\node [style=none] (7) at (-1.25, 0.125) {$\vdots$};
		\node [style=none] (8) at (-1.25, 0.75) {};
		\node [style=X phase dot] (9) at (0, 0) {$\alpha$};
	\end{pgfonlayer}
	\begin{pgfonlayer}{edgelayer}
		\draw [in=-180, out=-45] (0) to (2.center);
		\draw [in=180, out=45] (0) to (4.center);
		\draw [in=180, out=30, looseness=1.25] (0) to (1.center);
		\draw [in=135, out=0] (8.center) to (0);
		\draw [in=150, out=0] (5.center) to (0);
		\draw [in=-135, out=0] (6.center) to (0);
	\end{pgfonlayer}
\end{tikzpicture}

%% file: figures/def/def_hspider.tikz
\begin{tikzpicture}
	\begin{pgfonlayer}{nodelayer}
		\node [style=hadamard] (0) at (0, 0) {};
		\node [style=none] (1) at (-1.5, 0) {};
		\node [style=none] (2) at (1.5, 0) {};
	\end{pgfonlayer}
	\begin{pgfonlayer}{edgelayer}
		\draw (1.center) to (2.center);
	\end{pgfonlayer}
\end{tikzpicture}

%% file: figures/def/def_hadamard-edge1.tikz
\begin{tikzpicture}
	\begin{pgfonlayer}{nodelayer}
		\node [style=none] (0) at (-0.75, 0) {};
		\node [style=none] (1) at (0.75, 0) {};
		\node [style=none] (2) at (-1.25, 0) {};
		\node [style=Z dot] (3) at (-0.75, 0) {};
		\node [style=none] (6) at (1.25, 0) {};
		\node [style=Z dot] (7) at (0.75, 0) {};
	\end{pgfonlayer}
	\begin{pgfonlayer}{edgelayer}
		\draw [style=hadamard edge] (0.center) to (1.center);
		\draw (3) to (2.center);
		\draw (7) to (6.center);
	\end{pgfonlayer}
\end{tikzpicture}

%% file: figures/def/def_hadamard-edge0.tikz
\begin{tikzpicture}
	\begin{pgfonlayer}{nodelayer}
		\node [style=hadamard] (2) at (0, 0) {};
		\node [style=Z dot] (13) at (-0.75, 0) {};
		\node [style=none] (14) at (-1.25, 0) {};
		\node [style=Z dot] (18) at (0.75, 0) {};
		\node [style=none] (19) at (1.25, 0) {};
	\end{pgfonlayer}
	\begin{pgfonlayer}{edgelayer}
		\draw (13) to (14.center);
		\draw (18) to (19.center);
		\draw (13) to (18);
	\end{pgfonlayer}
\end{tikzpicture}

%% file: figures/measurement-effects/XY.tikz
\begin{tikzpicture}
	\begin{pgfonlayer}{nodelayer}
		\node [style=Z phase dot] (0) at (0.5, 0) {$\alpha$};
		\node [style=none] (1) at (-0.5, 0) {};
	\end{pgfonlayer}
	\begin{pgfonlayer}{edgelayer}
		\draw (0) to (1.center);
	\end{pgfonlayer}
\end{tikzpicture}

%% file: figures/measurement-effects/YZ.tikz
\begin{tikzpicture}
	\begin{pgfonlayer}{nodelayer}
		\node [style=none] (1) at (-0.5, 0) {};
		\node [style=X phase dot] (2) at (0.5, 0) {$\alpha$};
	\end{pgfonlayer}
	\begin{pgfonlayer}{edgelayer}
		\draw (2) to (1.center);
	\end{pgfonlayer}
\end{tikzpicture}

%% file: figures/measurement-effects/XZ.tikz
\begin{tikzpicture}
	\begin{pgfonlayer}{nodelayer}
		\node [style=none] (1) at (-1, 0) {};
		\node [style=X phase dot] (2) at (1, 0) {$\alpha$};
		\node [style=Z phase dot] (3) at (0, 0) {$\frac{\pi}{2}$};
	\end{pgfonlayer}
	\begin{pgfonlayer}{edgelayer}
		\draw (2) to (1.center);
	\end{pgfonlayer}
\end{tikzpicture}

%% file: figures/measurement-effects/X.tikz
\begin{tikzpicture}
	\begin{pgfonlayer}{nodelayer}
		\node [style=none] (1) at (-0.5, 0) {};
		\node [style=Z dot] (2) at (0.5, 0) {};
	\end{pgfonlayer}
	\begin{pgfonlayer}{edgelayer}
		\draw (2) to (1.center);
	\end{pgfonlayer}
\end{tikzpicture}

%% file: figures/measurement-effects/Y.tikz
\begin{tikzpicture}
	\begin{pgfonlayer}{nodelayer}
		\node [style=Z phase dot] (0) at (0.75, 0) {$-\frac{\pi}{2}$};
		\node [style=none] (1) at (-0.5, 0) {};
	\end{pgfonlayer}
	\begin{pgfonlayer}{edgelayer}
		\draw (0) to (1.center);
	\end{pgfonlayer}
\end{tikzpicture}

%% file: figures/measurement-effects/Z.tikz
\begin{tikzpicture}
	\begin{pgfonlayer}{nodelayer}
		\node [style=none] (1) at (-0.5, 0) {};
		\node [style=X dot] (2) at (0.5, 0) {};
	\end{pgfonlayer}
	\begin{pgfonlayer}{edgelayer}
		\draw (2) to (1.center);
	\end{pgfonlayer}
\end{tikzpicture}

%% file: figures/zx-rewrite-rules/0copy1.tikz
\begin{tikzpicture}
	\begin{pgfonlayer}{nodelayer}
		\node [style=Z phase dot] (0) at (0, 0) {$\alpha$};
		\node [style=none] (1) at (-1.5, 0) {};
		\node [style=none] (2) at (1, 0.75) {};
		\node [style=none] (3) at (1, -0.75) {};
		\node [style=none] (4) at (1.75, 0.25) {$\vdots$};
		\node [style=none] (6) at (1.75, 0.75) {};
		\node [style=none] (7) at (1.75, -0.75) {};
		\node [style=X dot] (8) at (-1.5, 0) {};
	\end{pgfonlayer}
	\begin{pgfonlayer}{edgelayer}
		\draw (0) to (1.center);
		\draw [in=-180, out=60] (0) to (2.center);
		\draw [in=180, out=-60] (0) to (3.center);
		\draw (6.center) to (2.center);
		\draw (7.center) to (3.center);
	\end{pgfonlayer}
\end{tikzpicture}

%% file: figures/zx-rewrite-rules/0copy2.tikz
\begin{tikzpicture}
	\begin{pgfonlayer}{nodelayer}
		\node [style=none] (2) at (1, 0.75) {};
		\node [style=none] (3) at (1, -0.75) {};
		\node [style=none] (4) at (2, 0.25) {$\vdots$};
		\node [style=none] (6) at (2, 0.75) {};
		\node [style=none] (7) at (2, -0.75) {};
		\node [style=X dot] (8) at (1, -0.75) {};
		\node [style=X dot] (9) at (1, 0.75) {};
	\end{pgfonlayer}
	\begin{pgfonlayer}{edgelayer}
		\draw (6.center) to (2.center);
		\draw (7.center) to (3.center);
	\end{pgfonlayer}
\end{tikzpicture}

%% file: figures/zx-rewrite-rules/fusion1.tikz
\begin{tikzpicture}
	\begin{pgfonlayer}{nodelayer}
		\node [style=Z phase dot] (3) at (-0.75, 0) {$\alpha$};
		\node [style=none] (4) at (-1.75, -0.75) {};
		\node [style=none] (5) at (-1.75, 0.75) {};
		\node [style=Z phase dot] (8) at (0.75, 0) {$\beta$};
		\node [style=none] (9) at (1.75, 0.75) {};
		\node [style=none] (10) at (1.75, -0.75) {};
		\node [style=none] (12) at (-1.75, 0.25) {$\vdots$};
		\node [style=none] (13) at (1.75, 0.25) {$\vdots$};
	\end{pgfonlayer}
	\begin{pgfonlayer}{edgelayer}
		\draw [in=0, out=-120] (3) to (4.center);
		\draw [in=0, out=120] (3) to (5.center);
		\draw [in=180, out=60] (8) to (9.center);
		\draw [in=180, out=-60] (8) to (10.center);
		\draw (3) to (8);
	\end{pgfonlayer}
\end{tikzpicture}

%% file: figures/zx-rewrite-rules/fusion2.tikz
\begin{tikzpicture}
	\begin{pgfonlayer}{nodelayer}
		\node [style=none] (4) at (-1.75, -0.75) {};
		\node [style=none] (5) at (-1.75, 0.75) {};
		\node [style=Z phase dot] (8) at (0, 0) {$\alpha+\beta$};
		\node [style=none] (9) at (1.75, 0.75) {};
		\node [style=none] (10) at (1.75, -0.75) {};
		\node [style=none] (12) at (-1.75, 0.25) {$\vdots$};
		\node [style=none] (13) at (1.75, 0.25) {$\vdots$};
	\end{pgfonlayer}
	\begin{pgfonlayer}{edgelayer}
		\draw [in=180, out=45] (8) to (9.center);
		\draw [in=180, out=-45] (8) to (10.center);
		\draw [in=0, out=-135] (8) to (4.center);
		\draw [in=0, out=135] (8) to (5.center);
	\end{pgfonlayer}
\end{tikzpicture}

%% file: figures/zx-rewrite-rules/color1.tikz
\begin{tikzpicture}
	\begin{pgfonlayer}{nodelayer}
		\node [style=none] (0) at (1.75, 0.75) {};
		\node [style=none] (1) at (1.75, -0.75) {};
		\node [style=Z phase dot] (2) at (0, 0) {$\alpha$};
		\node [style=none] (3) at (1.75, 0.25) {$\vdots$};
		\node [style=none] (4) at (-1.75, -0.75) {};
		\node [style=none] (5) at (-1.75, 0.75) {};
		\node [style=Z phase dot] (6) at (0, 0) {$\alpha$};
		\node [style=none] (7) at (-1.75, 0.25) {$\vdots$};
		\node [style=hadamard] (8) at (-1.25, 0.75) {};
		\node [style=hadamard] (9) at (1.25, 0.75) {};
		\node [style=hadamard] (10) at (1.25, -0.75) {};
		\node [style=hadamard] (11) at (-1.25, -0.75) {};
	\end{pgfonlayer}
	\begin{pgfonlayer}{edgelayer}
		\draw [in=-180, out=60] (2) to (0.center);
		\draw [in=180, out=-60] (2) to (1.center);
		\draw [in=0, out=-120] (6) to (4.center);
		\draw [in=0, out=120] (6) to (5.center);
	\end{pgfonlayer}
\end{tikzpicture}

%% file: figures/zx-rewrite-rules/color2.tikz
\begin{tikzpicture}
	\begin{pgfonlayer}{nodelayer}
		\node [style=none] (0) at (1.5, 0.75) {};
		\node [style=none] (1) at (1.5, -0.75) {};
		\node [style=X phase dot] (2) at (0, 0) {$\alpha$};
		\node [style=none] (3) at (1.5, 0.25) {$\vdots$};
		\node [style=none] (4) at (-1.5, -0.75) {};
		\node [style=none] (5) at (-1.5, 0.75) {};
		\node [style=X phase dot] (6) at (0, 0) {$\alpha$};
		\node [style=none] (7) at (-1.5, 0.25) {$\vdots$};
	\end{pgfonlayer}
	\begin{pgfonlayer}{edgelayer}
		\draw [in=-180, out=60] (2) to (0.center);
		\draw [in=180, out=-60] (2) to (1.center);
		\draw [in=0, out=-120] (6) to (4.center);
		\draw [in=0, out=120] (6) to (5.center);
	\end{pgfonlayer}
\end{tikzpicture}

%% file: figures/zx-rewrite-rules/binZ1.tikz
\begin{tikzpicture}
	\begin{pgfonlayer}{nodelayer}
		\node [style=none] (0) at (-1.5, 0) {};
		\node [style=none] (1) at (1.5, 0) {};
		\node [style=Z dot] (2) at (0, 0) {};
	\end{pgfonlayer}
	\begin{pgfonlayer}{edgelayer}
		\draw (1.center) to (0.center);
	\end{pgfonlayer}
\end{tikzpicture}

%% file: figures/zx-rewrite-rules/involh2.tikz
\begin{tikzpicture}
	\begin{pgfonlayer}{nodelayer}
		\node [style=none] (0) at (-1.5, 0) {};
		\node [style=none] (1) at (1.5, 0) {};
	\end{pgfonlayer}
	\begin{pgfonlayer}{edgelayer}
		\draw (1.center) to (0.center);
	\end{pgfonlayer}
\end{tikzpicture}

%% file: figures/zx-rewrite-rules/involh1.tikz
\begin{tikzpicture}
	\begin{pgfonlayer}{nodelayer}
		\node [style=none] (0) at (-1.5, 0) {};
		\node [style=none] (1) at (1.5, 0) {};
		\node [style=hadamard] (2) at (-0.5, 0) {};
		\node [style=hadamard] (3) at (0.5, 0) {};
	\end{pgfonlayer}
	\begin{pgfonlayer}{edgelayer}
		\draw (1.center) to (0.center);
	\end{pgfonlayer}
\end{tikzpicture}

%% file: figures/zx-rewrite-rules/copy1.tikz
\begin{tikzpicture}
	\begin{pgfonlayer}{nodelayer}
		\node [style=X phase dot] (0) at (0, 0) {$\alpha$};
		\node [style=none] (1) at (-1.75, 0) {};
		\node [style=none] (2) at (1, 0.75) {};
		\node [style=none] (3) at (1, -0.75) {};
		\node [style=none] (4) at (1.25, 0.25) {$\vdots$};
		\node [style=Z phase dot] (5) at (-1, 0) {$\pi$};
		\node [style=none] (6) at (1.5, 0.75) {};
		\node [style=none] (7) at (1.5, -0.75) {};
	\end{pgfonlayer}
	\begin{pgfonlayer}{edgelayer}
		\draw (0) to (1.center);
		\draw [in=-180, out=60] (0) to (2.center);
		\draw [in=180, out=-60] (0) to (3.center);
		\draw (6.center) to (2.center);
		\draw (7.center) to (3.center);
	\end{pgfonlayer}
\end{tikzpicture}

%% file: figures/zx-rewrite-rules/copy2.tikz
\begin{tikzpicture}
	\begin{pgfonlayer}{nodelayer}
		\node [style=X phase dot] (0) at (0, 0) {$-\alpha$};
		\node [style=none] (1) at (-1.25, 0) {};
		\node [style=none] (2) at (1, 0.75) {};
		\node [style=none] (3) at (1, -0.75) {};
		\node [style=none] (4) at (2, 0.25) {$\vdots$};
		\node [style=Z phase dot] (5) at (1.25, -0.75) {$\pi$};
		\node [style=none] (6) at (2, 0.75) {};
		\node [style=none] (7) at (2, -0.75) {};
		\node [style=Z phase dot] (8) at (1.25, 0.75) {$\pi$};
	\end{pgfonlayer}
	\begin{pgfonlayer}{edgelayer}
		\draw (0) to (1.center);
		\draw [in=-180, out=60] (0) to (2.center);
		\draw [in=180, out=-60] (0) to (3.center);
		\draw (6.center) to (2.center);
		\draw (7.center) to (3.center);
	\end{pgfonlayer}
\end{tikzpicture}

%% file: figures/graph-state/graph.tikz
\begin{tikzpicture}
	\begin{pgfonlayer}{nodelayer}
		\node [style=vertex] (0) at (1.25, 0.5) {};
		\node [style=vertex] (2) at (1.25, -0.5) {};
		\node [style=none] (4) at (1.25, 1) {$o_1$};
		\node [style=none] (7) at (1.25, -1) {$o_2$};
		\node [style=vertex] (8) at (0, 0.5) {};
		\node [style=vertex] (9) at (-1.25, 0.5) {};
		\node [style=vertex] (10) at (0, -0.5) {};
		\node [style=vertex] (11) at (-1.25, -0.5) {};
		\node [style=none] (12) at (0, 1) {$a$};
		\node [style=none] (13) at (-1.25, 1) {$i_1$};
		\node [style=none] (14) at (-1.25, -1) {$i_2$};
		\node [style=none] (15) at (0, -1) {$b$};
	\end{pgfonlayer}
	\begin{pgfonlayer}{edgelayer}
		\draw (0) to (8);
		\draw (8) to (9);
		\draw (2) to (10);
		\draw (10) to (11);
		\draw (2) to (8);
		\draw (8) to (11);
	\end{pgfonlayer}
\end{tikzpicture}

%% file: figures/graph-state/zx_graph.tikz
\begin{tikzpicture}
	\begin{pgfonlayer}{nodelayer}
		\node [style=Z dot] (0) at (-1.5, 0.75) {};
		\node [style=Z dot] (1) at (-1.5, -0.5) {};
		\node [style=Z dot] (3) at (0, -0.5) {};
		\node [style=Z dot] (4) at (0, 0.75) {};
		\node [style=Z dot] (5) at (1.5, 0.75) {};
		\node [style=Z dot] (6) at (1.5, -0.5) {};
		\node [style=none] (7) at (-2.5, 0.75) {};
		\node [style=none] (8) at (-2.5, -0.5) {};
		\node [style=none] (10) at (2.5, 0.75) {};
		\node [style=none] (11) at (2.5, -0.5) {};
		\node [style=none] (12) at (0.75, 1.25) {};
		\node [style=none] (13) at (-0.75, 1.25) {};
		\node [style=none] (14) at (-0.75, -1) {};
		\node [style=none] (15) at (0.75, -1) {};
	\end{pgfonlayer}
	\begin{pgfonlayer}{edgelayer}
		\draw [style=hadamard edge] (0) to (4);
		\draw [style=hadamard edge] (4) to (5);
		\draw [style=hadamard edge] (1) to (4);
		\draw [style=hadamard edge] (3) to (6);
		\draw [style=hadamard edge] (6) to (4);
		\draw (5) to (10.center);
		\draw (6) to (11.center);
		\draw (7.center) to (0);
		\draw (8.center) to (1);
		\draw [style=hadamard edge] (1) to (3);
		\draw (12.center) to (4);
		\draw (13.center) to (0);
		\draw (1) to (14.center);
		\draw (3) to (15.center);
	\end{pgfonlayer}
\end{tikzpicture}

%% file: figures/graph-state/ex-caus1.tikz
\begin{tikzpicture}
	\begin{pgfonlayer}{nodelayer}
		\node [style=Z dot] (0) at (-2, 1) {};
		\node [style=Z dot] (1) at (-2, -1) {};
		\node [style=Z dot] (3) at (0, -1) {};
		\node [style=Z dot] (4) at (0, 1) {};
		\node [style=Z dot] (5) at (2, 1) {};
		\node [style=Z dot] (6) at (2, -1) {};
		\node [style=none] (7) at (-3, 1) {};
		\node [style=none] (8) at (-3, -1) {};
		\node [style=none] (10) at (3, 1) {};
		\node [style=none] (11) at (3, -1) {};
		\node [style=none] (12) at (1, 1.75) {};
		\node [style=none] (13) at (-1, 1.75) {};
		\node [style=none] (14) at (-1, -1.75) {};
		\node [style=none] (15) at (1, -1.75) {};
		\node [style=Z phase dot] (18) at (-1, 1.75) {$\alpha_{i_1}$};
		\node [style=Z phase dot] (19) at (-1, -1.75) {$\alpha_{i_2}$};
		\node [style=none] (20) at (2.5, 0.5) {$o_1$};
		\node [style=none] (21) at (2.5, -1.5) {$o_2$};
		\node [style=none] (22) at (-2.5, 1.5) {$i_1$};
		\node [style=none] (23) at (-2.5, -1.5) {$i_2$};
		\node [style=Z phase dot] (24) at (1, 1.75) {$\alpha_a$};
		\node [style=Z phase dot] (25) at (1, -1.75) {$\alpha_b$};
		\node [style=none] (26) at (0, 1.75) {$a$};
		\node [style=none] (27) at (0, -1.75) {$b$};
	\end{pgfonlayer}
	\begin{pgfonlayer}{edgelayer}
		\draw [style=hadamard edge] (0) to (4);
		\draw [style=hadamard edge] (4) to (5);
		\draw [style=hadamard edge] (1) to (4);
		\draw [style=hadamard edge] (3) to (6);
		\draw [style=hadamard edge] (6) to (4);
		\draw (5) to (10.center);
		\draw (6) to (11.center);
		\draw (7.center) to (0);
		\draw (8.center) to (1);
		\draw [style=hadamard edge] (1) to (3);
		\draw (12.center) to (4);
		\draw (13.center) to (0);
		\draw (1) to (14.center);
		\draw (3) to (15.center);
	\end{pgfonlayer}
\end{tikzpicture}

%% file: figures/graph-state/ex-caus2.tikz
\begin{tikzpicture}
	\begin{pgfonlayer}{nodelayer}
		\node [style=Z dot] (0) at (-2, 1) {};
		\node [style=Z dot] (1) at (-2, -1) {};
		\node [style=Z dot] (3) at (0, -1) {};
		\node [style=Z dot] (4) at (0, 1) {};
		\node [style=Z dot] (5) at (2, 1) {};
		\node [style=Z dot] (6) at (2, -1) {};
		\node [style=none] (7) at (-3, 1) {};
		\node [style=none] (8) at (-3, -1) {};
		\node [style=none] (10) at (3, 1) {};
		\node [style=none] (11) at (3, -1) {};
		\node [style=none] (12) at (1, 1.75) {};
		\node [style=none] (13) at (-1, 1.75) {};
		\node [style=none] (14) at (-1, -1.75) {};
		\node [style=none] (15) at (1, -1.75) {};
		\node [style=Z phase dot] (18) at (-1, 1.75) {$\alpha_{i_1}$};
		\node [style=Z phase dot] (19) at (-1, -1.75) {$\alpha_{i_2}$};
		\node [style=none] (20) at (2.5, 0.5) {$o_1$};
		\node [style=none] (21) at (2.5, -1.5) {$o_2$};
		\node [style=none] (22) at (-2.5, 1.5) {$i_1$};
		\node [style=none] (23) at (-2.5, -1.5) {$i_2$};
		\node [style=Z dot] (26) at (-2.5, -0.25) {};
		\node [style=X dot] (27) at (-1.5, 0.5) {};
		\node [style=none] (28) at (-2.5, 0.25) {$z$};
		\node [style=Z phase dot] (29) at (1, 1.75) {$\alpha_a$};
		\node [style=Z phase dot] (30) at (1, -1.75) {$\alpha_b$};
		\node [style=none] (31) at (0, 1.75) {$a$};
		\node [style=none] (32) at (0, -1.75) {$b$};
	\end{pgfonlayer}
	\begin{pgfonlayer}{edgelayer}
		\draw [style=hadamard edge] (0) to (4);
		\draw [style=hadamard edge] (4) to (5);
		\draw [style=hadamard edge] (1) to (4);
		\draw [style=hadamard edge] (3) to (6);
		\draw [style=hadamard edge] (6) to (4);
		\draw (5) to (10.center);
		\draw (6) to (11.center);
		\draw (7.center) to (0);
		\draw (8.center) to (1);
		\draw [style=hadamard edge] (1) to (3);
		\draw (12.center) to (4);
		\draw (13.center) to (0);
		\draw (1) to (14.center);
		\draw (3) to (15.center);
		\draw [style=hadamard edge, in=15, out=-150] (4) to (26);
		\draw [style=hadamard edge, in=345, out=135] (3) to (26);
		\draw (26) to (27);
	\end{pgfonlayer}
\end{tikzpicture}

%% file: figures/graph-state/ex-caus3.tikz
\begin{tikzpicture}
	\begin{pgfonlayer}{nodelayer}
		\node [style=Z dot] (0) at (-2, 1) {};
		\node [style=Z dot] (1) at (-2, -1) {};
		\node [style=Z dot] (5) at (2, 1) {};
		\node [style=Z dot] (6) at (2, -1) {};
		\node [style=none] (7) at (-3, 1) {};
		\node [style=none] (8) at (-3, -1) {};
		\node [style=none] (10) at (3, 1) {};
		\node [style=none] (11) at (3, -1) {};
		\node [style=Z phase dot] (16) at (1, 1.75) {$\alpha_a$};
		\node [style=Z phase dot] (17) at (1, -1.75) {$\alpha_b$};
		\node [style=Z phase dot] (18) at (-1, 1.75) {$\alpha_{i_1}$};
		\node [style=Z phase dot] (19) at (-1, -1.75) {$\alpha_{i_2}$};
		\node [style=none] (20) at (2.5, 0.5) {$o_1$};
		\node [style=none] (21) at (2.5, -1.5) {$o_2$};
		\node [style=none] (22) at (-2.5, 1.5) {$i_1$};
		\node [style=none] (23) at (-2.5, -1.5) {$i_2$};
		\node [style=none] (24) at (0, 1.75) {$a$};
		\node [style=none] (25) at (0, -1.75) {$b$};
		\node [style=Z dot] (26) at (0, -1) {};
		\node [style=Z dot] (27) at (0, 1) {};
		\node [style=Z dot] (28) at (-2.75, -0.5) {};
		\node [style=X phase dot] (29) at (-1.75, 0.25) {$\alpha_z$};
		\node [style=none] (30) at (-3.25, -0.25) {$z$};
	\end{pgfonlayer}
	\begin{pgfonlayer}{edgelayer}
		\draw (5) to (10.center);
		\draw (6) to (11.center);
		\draw (7.center) to (0);
		\draw (8.center) to (1);
		\draw [style=hadamard edge, in=0, out=-135] (27) to (28);
		\draw [style=hadamard edge, in=-15, out=135] (26) to (28);
		\draw (28) to (29);
		\draw [style=hadamard edge] (0) to (27);
		\draw [style=hadamard edge] (27) to (5);
		\draw [style=hadamard edge, bend right=15] (1) to (27);
		\draw [style=hadamard edge] (26) to (6);
		\draw [style=hadamard edge] (6) to (27);
		\draw [style=hadamard edge] (1) to (26);
		\draw (16) to (27);
		\draw (18) to (0);
		\draw (1) to (19);
		\draw (26) to (17);
		\draw [style=hadamard edge, in=60, out=-120] (0) to (28);
	\end{pgfonlayer}
\end{tikzpicture}

%% file: figures/rmk-4-9/1.tikz
\begin{tikzpicture}
	\begin{pgfonlayer}{nodelayer}
		\node [style=Z dot] (1) at (0, 0.25) {};
		\node [style=Z dot] (2) at (1.75, 0.25) {};
		\node [style=Z phase dot] (4) at (1, 0.75) {$\alpha$};
		\node [style=Z dot] (5) at (1.75, -0.25) {};
		\node [style=none] (6) at (2.5, 0.25) {};
		\node [style=none] (7) at (2.5, -0.25) {};
		\node [style=none] (8) at (-0.5, -0.25) {$x$};
		\node [style=none] (9) at (3.25, 0.25) {$o_1$};
		\node [style=none] (10) at (3.25, -0.25) {$o_2$};
	\end{pgfonlayer}
	\begin{pgfonlayer}{edgelayer}
		\draw (7.center) to (5);
		\draw (6.center) to (2);
		\draw (4) to (1);
		\draw [style=hadamard edge] (2) to (1);
		\draw [style=hadamard edge] (1) to (5);
	\end{pgfonlayer}
\end{tikzpicture}

%% file: figures/rmk-4-9/2.tikz
\begin{tikzpicture}
	\begin{pgfonlayer}{nodelayer}
		\node [style=Z dot] (0) at (-1.5, 0.25) {};
		\node [style=Z dot] (1) at (0, 0.25) {};
		\node [style=Z dot] (2) at (1.75, 0.25) {};
		\node [style=Z phase dot] (4) at (1, 0.75) {$\alpha$};
		\node [style=Z dot] (5) at (1.75, -0.25) {};
		\node [style=none] (6) at (2.5, 0.25) {};
		\node [style=none] (7) at (2.5, -0.25) {};
		\node [style=none] (8) at (-0.5, -0.25) {$x$};
		\node [style=none] (9) at (3.25, 0.25) {$o_1$};
		\node [style=none] (10) at (3.25, -0.25) {$o_2$};
		\node [style=none] (11) at (-2, -0.25) {$z$};
		\node [style=X dot] (12) at (-0.5, 0.75) {};
	\end{pgfonlayer}
	\begin{pgfonlayer}{edgelayer}
		\draw (7.center) to (5);
		\draw (6.center) to (2);
		\draw (4) to (1);
		\draw [style=hadamard edge] (2) to (1);
		\draw [style=hadamard edge] (1) to (5);
		\draw [style=hadamard edge] (1) to (0);
		\draw (12) to (0);
	\end{pgfonlayer}
\end{tikzpicture}

%% file: figures/phase-shifting/1.tikz
\begin{tikzpicture}
	\begin{pgfonlayer}{nodelayer}
		\node [style=none] (3) at (-2, 0.75) {};
		\node [style=Z dot] (5) at (-1, 0) {};
		\node [style=none] (6) at (-2, -0.75) {};
		\node [style=Z dot] (7) at (-1, 0) {};
		\node [style=none] (8) at (0, 0.75) {};
		\node [style=Z phase dot] (9) at (0, 0.75) {$\alpha$};
		\node [style=none] (11) at (-2, 0.25) {$\vdots$};
		\node [style=Z dot] (12) at (1, 0) {};
		\node [style=none] (13) at (2, 0.75) {};
		\node [style=X dot] (14) at (2, 0.75) {};
	\end{pgfonlayer}
	\begin{pgfonlayer}{edgelayer}
		\draw [in=0, out=-120] (5) to (6.center);
		\draw [in=0, out=120] (5) to (3.center);
		\draw (7) to (8.center);
		\draw [style=hadamard edge] (7) to (12);
		\draw (13.center) to (12);
	\end{pgfonlayer}
\end{tikzpicture}

%% file: figures/phase-shifting/8.tikz
\begin{tikzpicture}
	\begin{pgfonlayer}{nodelayer}
		\node [style=none] (3) at (-2, 0.75) {};
		\node [style=Z dot] (5) at (-1, 0) {};
		\node [style=none] (6) at (-2, -0.75) {};
		\node [style=Z dot] (7) at (-1, 0) {};
		\node [style=none] (8) at (0, 0.75) {};
		\node [style=Z phase dot] (9) at (0, 0.75) {$\alpha-\alpha'$};
		\node [style=none] (11) at (-2, 0.25) {$\vdots$};
		\node [style=X phase dot] (16) at (2, 0.75) {$\alpha'$};
		\node [style=none] (17) at (1, 0) {};
		\node [style=Z dot] (18) at (1, 0) {};
	\end{pgfonlayer}
	\begin{pgfonlayer}{edgelayer}
		\draw [in=0, out=-120] (5) to (6.center);
		\draw [in=0, out=120] (5) to (3.center);
		\draw (7) to (8.center);
		\draw (16) to (18);
		\draw [style=hadamard edge] (7) to (18);
	\end{pgfonlayer}
\end{tikzpicture}

%% file: figures/phase-shifting/2.tikz
\begin{tikzpicture}
	\begin{pgfonlayer}{nodelayer}
		\node [style=none] (3) at (-2, 0.75) {};
		\node [style=Z dot] (5) at (-1, 0) {};
		\node [style=none] (6) at (-2, -0.75) {};
		\node [style=Z dot] (7) at (-1, 0) {};
		\node [style=none] (8) at (0, 0.75) {};
		\node [style=Z phase dot] (9) at (0, 0.75) {$\alpha$};
		\node [style=none] (11) at (-2, 0.25) {$\vdots$};
		\node [style=Z dot] (12) at (1, 0) {};
		\node [style=none] (13) at (2, 0.75) {};
		\node [style=X dot] (14) at (2, 0.75) {};
		\node [style=hadamard] (15) at (0, 0) {};
	\end{pgfonlayer}
	\begin{pgfonlayer}{edgelayer}
		\draw [in=0, out=-120] (5) to (6.center);
		\draw [in=0, out=120] (5) to (3.center);
		\draw (7) to (8.center);
		\draw (13.center) to (12);
		\draw (7) to (12);
	\end{pgfonlayer}
\end{tikzpicture}

%% file: figures/phase-shifting/3.tikz
\begin{tikzpicture}
	\begin{pgfonlayer}{nodelayer}
		\node [style=none] (3) at (-2, 0.75) {};
		\node [style=Z dot] (5) at (-1, 0) {};
		\node [style=none] (6) at (-2, -0.75) {};
		\node [style=Z dot] (7) at (-1, 0) {};
		\node [style=none] (8) at (-0.25, 1) {};
		\node [style=Z phase dot] (9) at (-0.25, 1) {$\alpha-\alpha'$};
		\node [style=none] (11) at (-2, 0.25) {$\vdots$};
		\node [style=Z dot] (12) at (2, 0) {};
		\node [style=none] (13) at (3, 0.75) {};
		\node [style=X dot] (14) at (3, 0.75) {};
		\node [style=hadamard] (15) at (1, 0) {};
		\node [style=Z phase dot] (16) at (0, 0) {$\alpha'$};
	\end{pgfonlayer}
	\begin{pgfonlayer}{edgelayer}
		\draw [in=0, out=-120] (5) to (6.center);
		\draw [in=0, out=120] (5) to (3.center);
		\draw (7) to (8.center);
		\draw (13.center) to (12);
		\draw (7) to (12);
	\end{pgfonlayer}
\end{tikzpicture}

%% file: figures/phase-shifting/4.tikz
\begin{tikzpicture}
	\begin{pgfonlayer}{nodelayer}
		\node [style=none] (3) at (-2, 0.75) {};
		\node [style=Z dot] (5) at (-1, 0) {};
		\node [style=none] (6) at (-2, -0.75) {};
		\node [style=Z dot] (7) at (-1, 0) {};
		\node [style=none] (8) at (0, 0.75) {};
		\node [style=Z phase dot] (9) at (0, 0.75) {$\alpha-\alpha'$};
		\node [style=none] (11) at (-2, 0.25) {$\vdots$};
		\node [style=Z dot] (12) at (2, 0) {};
		\node [style=none] (13) at (3, 0.75) {};
		\node [style=X dot] (14) at (3, 0.75) {};
		\node [style=hadamard] (15) at (0, 0) {};
		\node [style=X phase dot] (16) at (1, 0) {$\alpha'$};
	\end{pgfonlayer}
	\begin{pgfonlayer}{edgelayer}
		\draw [in=0, out=-120] (5) to (6.center);
		\draw [in=0, out=120] (5) to (3.center);
		\draw (7) to (8.center);
		\draw (13.center) to (12);
		\draw (7) to (12);
	\end{pgfonlayer}
\end{tikzpicture}

%% file: figures/phase-shifting/5.tikz
\begin{tikzpicture}
	\begin{pgfonlayer}{nodelayer}
		\node [style=none] (3) at (-2, 0.75) {};
		\node [style=Z dot] (5) at (-1, 0) {};
		\node [style=none] (6) at (-2, -0.75) {};
		\node [style=Z dot] (7) at (-1, 0) {};
		\node [style=none] (8) at (0, 0.75) {};
		\node [style=Z phase dot] (9) at (0, 0.75) {$\alpha-\alpha'$};
		\node [style=none] (11) at (-2, 0.25) {$\vdots$};
		\node [style=X dot] (14) at (2, 0) {};
		\node [style=hadamard] (15) at (0, 0) {};
		\node [style=X phase dot] (16) at (1, 0) {$\alpha'$};
		\node [style=none] (17) at (2, 0) {};
	\end{pgfonlayer}
	\begin{pgfonlayer}{edgelayer}
		\draw [in=0, out=-120] (5) to (6.center);
		\draw [in=0, out=120] (5) to (3.center);
		\draw (7) to (8.center);
		\draw (17.center) to (7);
	\end{pgfonlayer}
\end{tikzpicture}

%% file: figures/phase-shifting/6.tikz
\begin{tikzpicture}
	\begin{pgfonlayer}{nodelayer}
		\node [style=none] (3) at (-2, 0.75) {};
		\node [style=Z dot] (5) at (-1, 0) {};
		\node [style=none] (6) at (-2, -0.75) {};
		\node [style=Z dot] (7) at (-1, 0) {};
		\node [style=none] (8) at (0, 0.75) {};
		\node [style=Z phase dot] (9) at (0, 0.75) {$\alpha-\alpha'$};
		\node [style=none] (11) at (-2, 0.25) {$\vdots$};
		\node [style=hadamard] (15) at (0, 0) {};
		\node [style=X phase dot] (16) at (1, 0) {$\alpha'$};
		\node [style=none] (17) at (1, 0) {};
	\end{pgfonlayer}
	\begin{pgfonlayer}{edgelayer}
		\draw [in=0, out=-120] (5) to (6.center);
		\draw [in=0, out=120] (5) to (3.center);
		\draw (7) to (8.center);
		\draw (17.center) to (7);
	\end{pgfonlayer}
\end{tikzpicture}

%% file: figures/phase-shifting/7.tikz
\begin{tikzpicture}
	\begin{pgfonlayer}{nodelayer}
		\node [style=none] (3) at (-2, 0.75) {};
		\node [style=Z dot] (5) at (-1, 0) {};
		\node [style=none] (6) at (-2, -0.75) {};
		\node [style=Z dot] (7) at (-1, 0) {};
		\node [style=none] (8) at (0, 0.75) {};
		\node [style=Z phase dot] (9) at (0, 0.75) {$\alpha-\alpha'$};
		\node [style=none] (11) at (-2, 0.25) {$\vdots$};
		\node [style=hadamard] (15) at (0, 0) {};
		\node [style=X phase dot] (16) at (2, 0.75) {$\alpha'$};
		\node [style=none] (17) at (1, 0) {};
		\node [style=Z dot] (18) at (1, 0) {};
	\end{pgfonlayer}
	\begin{pgfonlayer}{edgelayer}
		\draw [in=0, out=-120] (5) to (6.center);
		\draw [in=0, out=120] (5) to (3.center);
		\draw (7) to (8.center);
		\draw (17.center) to (7);
		\draw (16) to (18);
	\end{pgfonlayer}
\end{tikzpicture}

%% file: figures/vertex-splitting/1.tikz
\begin{tikzpicture}
	\begin{pgfonlayer}{nodelayer}
		\node [style=none] (3) at (-2, 0.75) {};
		\node [style=Z dot] (5) at (-1, 0) {};
		\node [style=none] (6) at (-2, -0.75) {};
		\node [style=Z dot] (7) at (-1, 0) {};
		\node [style=none] (8) at (0, 0.75) {};
		\node [style=Z phase dot] (9) at (0, 0.75) {$\alpha$};
		\node [style=none] (10) at (-3, 0) {$N(x)$};
		\node [style=none] (11) at (-2, 0.25) {$\vdots$};
		\node [style=none] (12) at (-0.5, -0.5) {$x$};
	\end{pgfonlayer}
	\begin{pgfonlayer}{edgelayer}
		\draw [in=0, out=-120] (5) to (6.center);
		\draw [in=0, out=120] (5) to (3.center);
		\draw (7) to (8.center);
	\end{pgfonlayer}
\end{tikzpicture}

%% file: figures/vertex-splitting/2.tikz
\begin{tikzpicture}
	\begin{pgfonlayer}{nodelayer}
		\node [style=none] (3) at (-2, 0.75) {};
		\node [style=Z dot] (5) at (-1, 0) {};
		\node [style=none] (6) at (-2, -0.75) {};
		\node [style=Z dot] (7) at (-1, 0) {};
		\node [style=none] (8) at (0, 0.75) {};
		\node [style=Z phase dot] (9) at (0, 0.75) {$\alpha-\alpha'$};
		\node [style=none] (10) at (-3, 0) {$N(x)$};
		\node [style=none] (11) at (-2, 0.25) {$\vdots$};
		\node [style=X phase dot] (16) at (2, 0.75) {$\alpha'$};
		\node [style=none] (17) at (1, 0) {};
		\node [style=Z dot] (18) at (1, 0) {};
		\node [style=none] (19) at (-0.5, -0.5) {$x$};
		\node [style=none] (20) at (1.5, -0.5) {$x''$};
	\end{pgfonlayer}
	\begin{pgfonlayer}{edgelayer}
		\draw [in=0, out=-120] (5) to (6.center);
		\draw [in=0, out=120] (5) to (3.center);
		\draw (7) to (8.center);
		\draw (16) to (18);
		\draw [style=hadamard edge] (7) to (18);
	\end{pgfonlayer}
\end{tikzpicture}

%% file: figures/vertex-splitting/3.tikz
\begin{tikzpicture}
	\begin{pgfonlayer}{nodelayer}
		\node [style=none] (3) at (-2, 2) {};
		\node [style=Z dot] (5) at (-1, 1.25) {};
		\node [style=none] (6) at (-2, 0.5) {};
		\node [style=Z dot] (7) at (-1, 1.25) {};
		\node [style=none] (8) at (-0.25, 2.25) {};
		\node [style=Z phase dot] (9) at (-0.25, 2.25) {$\alpha-\alpha'$};
		\node [style=none] (10) at (-3, 1.25) {$N(x)$};
		\node [style=none] (11) at (-2, 1.5) {$\vdots$};
		\node [style=X phase dot] (16) at (1.75, 2.25) {$\alpha'$};
		\node [style=none] (17) at (1, 1.25) {};
		\node [style=Z dot] (18) at (1, 1.25) {};
		\node [style=Z dot] (19) at (2, 0) {};
		\node [style=none] (20) at (-2, -0.5) {};
		\node [style=none] (21) at (-2, -2) {};
		\node [style=none] (22) at (-2, -1) {$\vdots$};
		\node [style=none] (23) at (-2.75, -1) {$W$};
		\node [style=none] (24) at (2.75, 1) {};
		\node [style=X dot] (25) at (2.75, 1) {};
		\node [style=none] (26) at (-0.25, 0.5) {$x$};
		\node [style=none] (27) at (2, 1.5) {$x''$};
		\node [style=none] (28) at (2.5, -0.75) {$x'$};
	\end{pgfonlayer}
	\begin{pgfonlayer}{edgelayer}
		\draw [in=0, out=-120] (5) to (6.center);
		\draw [in=0, out=120] (5) to (3.center);
		\draw (7) to (8.center);
		\draw (16) to (18);
		\draw [style=hadamard edge] (7) to (18);
		\draw [style=hadamard edge, in=105, out=-15] (18) to (19);
		\draw [in=-165, out=0] (20.center) to (19);
		\draw [in=0, out=-120] (19) to (21.center);
		\draw (24.center) to (19);
	\end{pgfonlayer}
\end{tikzpicture}

%% file: figures/vertex-splitting/4.tikz
\begin{tikzpicture}
	\begin{pgfonlayer}{nodelayer}
		\node [style=none] (3) at (-2, 2) {};
		\node [style=Z dot] (5) at (-1, 1.25) {};
		\node [style=none] (6) at (-2, 0.5) {};
		\node [style=Z dot] (7) at (-1, 1.25) {};
		\node [style=none] (8) at (-0.25, 2.25) {};
		\node [style=Z phase dot] (9) at (-0.25, 2.25) {$\alpha-\alpha'$};
		\node [style=none] (10) at (-3.75, 1.25) {$W\symd N(x) $};
		\node [style=none] (11) at (-2, 1.5) {$\vdots$};
		\node [style=none] (17) at (1, 1.25) {};
		\node [style=Z dot] (18) at (1, 1.25) {};
		\node [style=Z dot] (19) at (2, 0) {};
		\node [style=none] (20) at (-2, -0.5) {};
		\node [style=none] (21) at (-2, -2) {};
		\node [style=none] (22) at (-2, -1) {$\vdots$};
		\node [style=none] (23) at (-2.75, -1) {$W$};
		\node [style=none] (24) at (2.75, 1) {};
		\node [style=none] (26) at (1.75, 2.25) {};
		\node [style=Z phase dot] (27) at (2.75, 1) {$\alpha'$};
		\node [style=Z dot] (28) at (1.75, 2.25) {};
		\node [style=none] (29) at (-0.5, 0.5) {$x$};
		\node [style=none] (30) at (1.75, 1.5) {$x'$};
		\node [style=none] (31) at (2.25, -0.75) {$x''$};
	\end{pgfonlayer}
	\begin{pgfonlayer}{edgelayer}
		\draw [in=0, out=-120] (5) to (6.center);
		\draw [in=0, out=120] (5) to (3.center);
		\draw (7) to (8.center);
		\draw [style=hadamard edge] (7) to (18);
		\draw [style=hadamard edge, in=105, out=-15] (18) to (19);
		\draw [in=-165, out=0] (20.center) to (19);
		\draw [in=0, out=-120] (19) to (21.center);
		\draw (24.center) to (19);
		\draw (26.center) to (18);
	\end{pgfonlayer}
\end{tikzpicture}

%% file: figures/pivoting/1.tikz
\begin{tikzpicture}
	\begin{pgfonlayer}{nodelayer}
		\node [style=Z dot] (2) at (-0.5, -0.5) {};
		\node [style=Z dot] (3) at (-1.5, 2) {};
		\node [style=none] (5) at (-0.5, 2.75) {};
		\node [style=none] (34) at (2.25, -2.75) {};
		\node [style=Z dot] (41) at (2, 1.5) {};
		\node [style=Z dot] (42) at (-0.5, 0.5) {};
		\node [style=none] (43) at (3, 2.25) {};
		\node [style=none] (44) at (2.5, 1.25) {$u$};
		\node [style=none] (45) at (2.5, -1.25) {$v$};
		\node [style=none] (48) at (2.25, 2.75) {};
		\node [style=Z dot] (51) at (-3, 2) {};
		\node [style=none] (52) at (-2, 2.75) {};
		\node [style=Z dot] (53) at (-1.5, -2) {};
		\node [style=none] (54) at (-0.5, -2.75) {};
		\node [style=Z dot] (55) at (2, -1.5) {};
		\node [style=none] (56) at (3, -2.25) {};
		\node [style=Z dot] (57) at (-3, -2) {};
		\node [style=none] (58) at (-2, -2.75) {};
	\end{pgfonlayer}
	\begin{pgfonlayer}{edgelayer}
		\draw (3) to (5.center);
		\draw (2) to (34.center);
		\draw (41) to (43.center);
		\draw [style=hadamard edge] (41) to (42);
		\draw (42) to (48.center);
		\draw (41) to (43.center);
		\draw [style=hadamard edge] (41) to (2);
		\draw [style=hadamard edge, in=135, out=15] (3) to (41);
		\draw (51) to (52.center);
		\draw [style=hadamard edge, in=150, out=-45] (51) to (41);
		\draw (53) to (54.center);
		\draw (55) to (56.center);
		\draw (55) to (56.center);
		\draw [style=hadamard edge, in=-135, out=-15] (53) to (55);
		\draw (57) to (58.center);
		\draw [style=hadamard edge, in=-150, out=45] (57) to (55);
		\draw (55) to (56.center);
		\draw [style=hadamard edge] (41) to (55);
		\draw [style=hadamard edge] (42) to (55);
		\draw (55) to (56.center);
		\draw [style=hadamard edge] (41) to (55);
		\draw [style=hadamard edge] (2) to (55);
		\draw [style=hadamard edge] (3) to (42);
		\draw [style=hadamard edge, in=60, out=-135] (42) to (57);
		\draw [style=hadamard edge] (2) to (53);
		\draw [style=hadamard edge, in=135, out=-60] (51) to (2);
		\draw [style=hadamard edge, bend right=15] (51) to (57);
		\draw [style=hadamard edge, bend right=15] (3) to (53);
	\end{pgfonlayer}
\end{tikzpicture}

%% file: figures/pivoting/2.tikz
\begin{tikzpicture}
	\begin{pgfonlayer}{nodelayer}
		\node [style=Z dot] (2) at (-0.5, -0.5) {};
		\node [style=Z dot] (3) at (-1.5, 2) {};
		\node [style=none] (5) at (-0.5, 2.75) {};
		\node [style=none] (34) at (2.25, -2.75) {};
		\node [style=Z dot] (42) at (-0.5, 0.5) {};
		\node [style=none] (43) at (3.75, -1.5) {};
		\node [style=none] (44) at (4.75, 1.25) {$u$};
		\node [style=none] (45) at (4.75, -1.25) {$v$};
		\node [style=none] (48) at (2.25, 2.75) {};
		\node [style=Z dot] (51) at (-3, 2) {};
		\node [style=none] (52) at (-2, 2.75) {};
		\node [style=Z dot] (53) at (-1.5, -2) {};
		\node [style=none] (54) at (-0.5, -2.75) {};
		\node [style=Z dot] (57) at (-3, -2) {};
		\node [style=none] (58) at (-2, -2.75) {};
		\node [style=Z dot] (59) at (2, -1.5) {};
		\node [style=none] (60) at (3.75, 1.5) {};
		\node [style=Z dot] (61) at (2, 1.5) {};
		\node [style=none] (62) at (4.5, 1.5) {};
		\node [style=none] (63) at (4.5, -1.5) {};
		\node [style=hadamard] (64) at (4, 1.5) {};
		\node [style=hadamard] (65) at (4, -1.5) {};
		\node [style=Z phase dot] (66) at (0.75, -1.5) {$\pi$};
		\node [style=Z phase dot] (67) at (0.75, 1.5) {$\pi$};
		\node [style=none] (69) at (4.5, 1.5) {};
		\node [style=none] (70) at (5.25, 2) {};
		\node [style=none] (71) at (4.5, -1.5) {};
		\node [style=none] (72) at (4.5, -1.5) {};
		\node [style=none] (73) at (5.25, -2) {};
	\end{pgfonlayer}
	\begin{pgfonlayer}{edgelayer}
		\draw (3) to (5.center);
		\draw (2) to (34.center);
		\draw (42) to (48.center);
		\draw (51) to (52.center);
		\draw (53) to (54.center);
		\draw (57) to (58.center);
		\draw [in=180, out=15, looseness=0.75] (59) to (60.center);
		\draw [style=hadamard edge] (61) to (42);
		\draw [in=-180, out=-15, looseness=0.75] (61) to (43.center);
		\draw [style=hadamard edge] (61) to (2);
		\draw [style=hadamard edge, in=135, out=15] (3) to (61);
		\draw [style=hadamard edge, in=150, out=-45] (51) to (61);
		\draw [style=hadamard edge, in=-135, out=-15] (53) to (59);
		\draw [style=hadamard edge, in=-150, out=45] (57) to (59);
		\draw [style=hadamard edge] (42) to (59);
		\draw [style=hadamard edge] (2) to (59);
		\draw (63.center) to (43.center);
		\draw (60.center) to (62.center);
		\draw [style=hadamard edge, in=90, out=-120] (3) to (57);
		\draw [style=hadamard edge, in=120, out=-90] (51) to (53);
		\draw [style=hadamard edge, in=165, out=60] (57) to (2);
		\draw [style=hadamard edge, bend left=15] (2) to (3);
		\draw [style=hadamard edge, in=-165, out=-60] (51) to (42);
		\draw [style=hadamard edge, bend right=15] (42) to (53);
		\draw [in=0, out=-135] (70.center) to (69.center);
		\draw [in=0, out=135] (73.center) to (72.center);
		\draw [style=hadamard edge] (61) to (59);
	\end{pgfonlayer}
\end{tikzpicture}

%% file: figures/pivoting/5.tikz
\begin{tikzpicture}
	\begin{pgfonlayer}{nodelayer}
		\node [style=none] (1) at (-0.75, 0) {};
		\node [style=X phase dot] (2) at (0.75, 0) {$\alpha$};
		\node [style=hadamard] (3) at (-0.25, 0) {};
	\end{pgfonlayer}
	\begin{pgfonlayer}{edgelayer}
		\draw (2) to (1.center);
	\end{pgfonlayer}
\end{tikzpicture}

%% file: figures/pivoting/6.tikz
\begin{tikzpicture}
	\begin{pgfonlayer}{nodelayer}
		\node [style=Z phase dot] (0) at (0.5, 0) {$\alpha$};
		\node [style=none] (1) at (-0.5, 0) {};
	\end{pgfonlayer}
	\begin{pgfonlayer}{edgelayer}
		\draw (0) to (1.center);
	\end{pgfonlayer}
\end{tikzpicture}